\documentclass[superscriptaddress,twocolumn, showpacs,nofootinbib,longbibliography,notitlepage]{revtex4-2}
\UseRawInputEncoding
\usepackage{pgfplots}
\usepackage{etex}
\usepackage{physics}
\usepackage{amsmath,amssymb,amsthm,amstext,mathrsfs}
\usepackage{times,txfonts}
\usepackage{braket}
\usepackage{subcaption}
\usepackage{color}
\usepackage{natbib}
\usepackage{blkarray}
\usepackage{mathtools}
\usepackage{latexsym}
\usepackage{tabularx, booktabs}
\usepackage{graphics,epstopdf}
\usepackage{float}
\usepackage{graphicx}
\usepackage{amsfonts}
\usepackage{subcaption}
\usepackage{enumerate}
\usepackage{soul}
\usepackage{comment}
\usepackage[inkscapelatex=true]{svg}
\usepackage{hyperref}
\hypersetup{
	colorlinks=true,
	linkcolor=red,
	filecolor=blue,      
	urlcolor=blue,
	citecolor=purple,
}

\newtheorem{Defi}{Definition}
\newtheorem{Propos}{Proposition}

\newtheorem{thm}{Theorem}
\newtheorem{Lemma}{Lemma}
\newtheorem{Remark}{Remark}

\newtheorem{Corollary}{Corollary}

\begin{document}

\title{Unbounded Communication Power of a Qubit}

\author{Souradeep Sasmal}
\email{souradeep.007@gmail.com}
\affiliation{Institute of Fundamental and Frontier Sciences, University of Electronic Science and Technology of China, Chengdu 611731, China}

\author{Som Kanjilal}
%\email{}
\affiliation{International Iberian Nanotechnology Laboratory, Avenida Mestre Jos\'{e} Veiga, 4715-330 Braga, Portugal}

\author{Debarshi Das}
\email{dasdebarshi90@gmail.com}
\affiliation{Department of Physics, Shiv Nadar Institution of Eminence, Gautam Buddha Nagar, Uttar Pradesh 201314, India}

%%%%%%%%%%%%%%%%%%%%%%%%%%%%%%%%%%%%%%%%%%%%%%%%%%%%%%%%%%%%%%%%%%%%%%%%%%%%%%%%%%%%%%%%%%%%%%%%%%%%%%%%%%%%%%%%%%%%%%%%%%%%%%%%%%%%%%%%%%%%%%%%%%%%%%%%%%%%%%%%%%%%%%%%%%%%%%%%%%%%%%%%%%%%%%%%%%%%%%%%%%%%%%%%%%%%%%%%%%%%%%%%%%%%%%%%%%%%%%%%%%%%%%%%%%%%%%%%%%%%%%%%%%%%%%%%%%%%%%%%%%%%%%%%%%%%%%%%%%%%%%%%%%%%%%%%%%%%%%%%%%%%%%%%%%%%%%%%%%%%%%%%%%%%%%%%

\begin{abstract}

Quantum mechanics enables information-processing advantages even at the level of a single qubit. A paradigmatic example is the 2$\to$1 random access code (RAC), where a qubit outperforms a classical bit in retrieving encoded information. In the standard form, however, this quantum advantage is restricted to a single receiver, since decoding measurements inevitably destroy the encoded information. Contrary to this, we address how long the information encoded in a single qubit remains accessible even after multiple decoding, each with a quantum advantage. Introducing preparation distinguishability as an operational resource associated with the sender, we show that its interplay with measurement incompatibility on the receiver's side can mitigate measurement-induced disturbance, thereby enabling an arbitrarily long sequence of receivers to each retain a quantum advantage. Our results show that, even under repeated measurements, the information encoded in a qubit need not be entirely exhausted, revealing a stronger communication feature than previously recognised.

\end{abstract}
 

\maketitle

%%%%%%%%%%%%%%%%%%%%%%%%%%%%%%%%%%%%%%%%%%%%%%%%%%%%%%%%%%%%%%%%%%%%%%%%%%%%%%%%%%%%%%%%%%%%%%%%%%%%%%%%%%%%%%%%%%%%%%%%%%%%%%%%%%%%%%%%%%%%%%%%%%%%%%%%%%%%%%%%%%%%%%%%%%%%%%%%%%%%%%%%%%%%%%%%%%%%%%%%%%%%%%%%%%%%%%%%%%%%%%%%%%%%%%%%%%%%%%%%%%%%%%%%%%%%%%%%%%%%%%%%%%%%%%%%%%%%%%%%%%%%%%%%%%%%%%%%%%%%%%%%%%%%%%%%%%%%%%%%%%%%%%%%%%%%%%%%%%%%%%%%%%%%%%%%

\section{Introduction}

The ability of a single qubit to outperform a single classical bit in information-processing tasks is a cornerstone of quantum information theory \cite{Weisner1983, Bennett1992}. For instance, in the 2$\to$1 random access code (RAC), a sender, Alice, encodes two classical bits $x=(x_1,x_2)\in\qty{0, 1}^{2}$ into a single qubit, enabling a receiver, Bob, to retrieve any one of them with a success probability exceeding that with one classical bit of communication \cite{Ambainis1999,Ambainis2002,Ambainis2009, Spekkens2009}. This advantage is widely regarded as a manifestation of the intrinsic power of quantum systems in prepare-measure based communication tasks \cite{Gallego2010,Gois2021,Manna2024,Singh2025,Brask2026}. 

Since quantum measurements are inherently invasive \cite{Home1997book}, the standard formulation of quantum RACs typically assumes that the receiver performs projective measurements for decoding, thereby effectively destroying the accessible information encoded in the qubit. However, recent work has demonstrated that this conventional paradigm does not fully exploit the communication power of a qubit: by employing non-projective (unsharp) measurements \cite{Busch1986, Busch2016book}, it becomes possible for up to two sequential independent receivers to extract information, each establishing a quantum advantage \cite{Mohan2019,Anwer2020,Milkin2020}. This naturally raises a fundamental question: is this the ultimate limit of a qubit's sequential communication capability? 

Specifically, we consider a sequential RAC scenario comprising  a single Alice, and multiple receivers, Bob$^{(k)}$ with $k=1,2,3, \ldots$, who act sequentially and independently. In each run, Alice prepares a qubit state $\{\rho^{(k=1)}_{x_1x_2}\}$ encoding her input and sends it to Bob$^{(k=1)}$. Upon receiving an input $y_{k=1}\in\{1,2\}$, Bob$^{(1)}$ performs one of two dichotomic qubit measurements in order to guess the $y_1$-th bit. The post-measurement state is then passed to Bob$^{(2)}$, who performs an analogous decoding procedure upon receiving $y_2\in\{1,2\}$. This process continues for an arbitrary number of sequential receivers. Crucially, all Bobs act independently and have no access to the inputs or outcomes of the others. The scenario is illustrated in Fig.~\ref{seq}. The central question we address is: how many sequential receivers can each obtain a quantum advantage?

We address this question by relating the standard 2$\to$1 RAC to two constrained preparation discrimination tasks, whose optimal success probabilities are quantified by \textit{preparation distinguishability}, viz. the maximal statistical distance between marginal ensembles prepared by the sender. We identify this quantity as a necessary resource for achieving quantum advantage. On the receiver's side, measurement incompatibility \cite{Heinosaari2016,Gunhe2023}, quantified through the unsharpness parameters in the context of unsharp qubit measurements \cite{Busch1986, Busch2016book}, acts as a complementary resource \cite{Saha2023,Saha2020,Carmeli2020,Skrzypczyk2019}. We further establish a trade-off relation between these two resources in determining the quantum advantage.

Building on this framework, we analyse different decoding strategies. When both measurements are unsharp with the same unsharpness parameter, as considered in \cite{Mohan2019,Anwer2020}, preparation distinguishabilities approaching the classical limit, i.e. one pair of marginal ensembles becoming perfectly distinguishable and the other pair becoming indistinguishable, requires a greater degree of quantumness on the receiver's side, i.e. an increased degree of measurement incompatibility, in order to sustain a quantum advantage.

In contrast, when one measurement is projective and the other unsharp, a reduced measurement incompatibility suffices to retain quantum advantage, even when the preparation distinguishabilities approach the classical bound. This leads to a situation in which such decoding measurements preserve the preparation distinguishabilities maximally. This observation underpins our main result: despite the unavoidable disturbance induced by each measurement, an arbitrarily long sequence of independent receivers can, in principle, extract information from a single qubit, each with a quantum advantage in the sequential RAC scenario. We begin by formalising the notion of preparation distinguishability.

\begin{figure}[t!]
\includegraphics[width=0.85\linewidth]{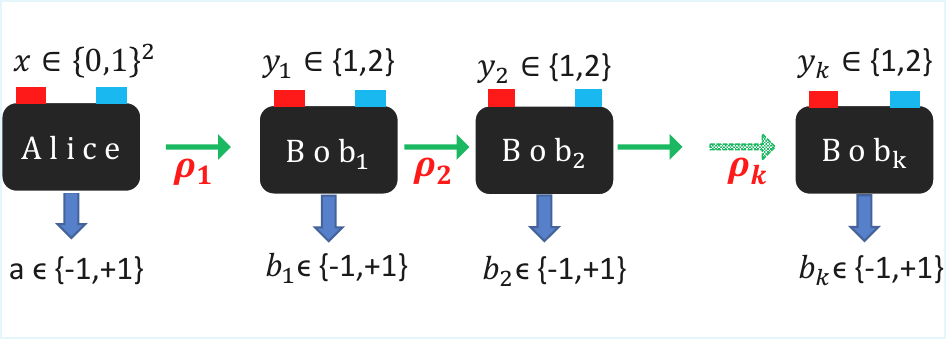}
\caption{Alice's preparations are sequentially passed from one Bob to the next, each of whom performs one of the two measurements aiming to independently reveal quantum advantage in the 2$\to$1 RAC.}
\label{seq}
\end{figure}

%%%%%%%%%%%%%%%%%%%%%%%%%%%%%%%%%%%%%%%%%%%%%%%%%%%%%%%%%%%%%%%%%%%%%%%%%%%%%%%%%%%%%%%%%%%%%%%%%%%%%%%%%%%%%%%%%%%%%%%%%%%%%%%%%%%%%%%%%%%%%%%%%%%%%%%%%%%%%%%%%%%%%%%%%%%%%%%%%%%%%%%%%%%%%%%%%%%%%%%%%%%%%%%%%%%%%%%%%%%%%%%%%%%%%%%%%%%%%%%%%%%%%%%%%%%%%%%%%%%%%%%%%%%%%%%%%%%%%%%%%%%%%%%%%%%%%%%%%%%%%%%%%%%%%%%%%%%%%%%%%%%%%%%%%%%%%%%%%%%%%%%%%%%%%%%%

\section{Role of preparation distinguishabilities in 2$\to$1 RAC}

\subsection{Preparation Distinguishability} 

In any operational theory, a \textit{preparation} $P\in\mathscr{P}$ specifies a procedure for generating physical systems according to some ensemble description, while a \textit{measurement} $M \in \mathscr{M}$ is a test applied to those systems, producing outcomes $m\in\mathcal{O}_m$ with probability $p(m|P,M)$ i.e.  $p(m|P,M)\geq 0$ for all $m$ and $\sum_m  p(m|P,M) =1$ \cite{Spekkens2005}. 
\begin{Defi}
The \textit{preparation distinguishability} between the two preparations $P_0$ and $P_1$ is the maximum total variation distance between their outcome distributions, optimised over all possible measurements
\begin{equation}\label{pdisd}
    \Delta(P_0, P_1) = \max_M \ \frac{1}{2}\sum_m \abs{p(m|P_0,M)-p(m|P_1,M)}.
\end{equation}
\end{Defi}
By definition, $0\leq \Delta(P_0, P_1) \leq 1$. If $\Delta(P_0, P_1)=0$, the preparations are operationally equivalent, since no measurement can distinguish between them \cite{Schmid2018}. If $ \Delta(P_0, P_1)=1$, their outcome distributions have disjoint support, i.e. there exists a measurement that perfectly distinguishes them. 

To understand its operational significance, consider a binary discrimination task in which one of $P_0$ or $P_1$ is chosen uniformly and given to Bob, who performs a measurement $M$ with outcomes $m\in\qty{0,1}$. Bob guesses $P_0$ upon observing $m=0$, and  $P_1$ otherwise. His average success probability, optimised over all measurements, is 
\begin{equation}\label{optgdistas}
    p_{\rm g}^{\rm opt}(P_0,P_1)= \frac{1}{2}  \max_M \sum_{m\in \{0,1\}} p(m|P_m,M).
\end{equation}
For any pair of preparations $P_0,P_1$ (see Appendix \ref{apcgppd}),    
\begin{equation}\label{cgppdeq}
p_{\rm g}^{ \rm opt}(P_0,P_1)=\frac{1}{2}\qty[1+\Delta(P_0,P_1)].
\end{equation}
This operational equivalence between $p_{\rm g}^{\rm opt}(P_0,P_1)$ and $\Delta(P_0,P_1)$ will play a central role in what follows.

%%%%%%%%%%%%%%%%%%%%%%%%%%%%%%%%%%%%%%%%%%%%%%%%%%%%%%%%%%%%%%%%%%%%%%%%%%%%%%%%%%%%%%%%%%%%%%%%%%%%%%%%%%%%%%%%%%%%%%%%%%%%%%%%%%%%%%%%%%%%%%%%%%%%%%%%%%%%%%%%%%%%%%%%%%%%%%%%%%%%%%%%%%%%%%%%%%%%%%%%%%%%%%%%%%%%%%%%%%%%%%%%

\subsection{2$\to$1 RAC task as preparation discrimination tasks}

After receiving the encoded preparation $P_{x} \in \{P_{x_1x_2}\}$ from Alice via one classical-/quantum-bit  communication, Bob performs a dichotomic measurement $M_y$, depending on his input $y\in\{1,2\}$, yielding an outcome $b\in\qty{0,1}$. The task is successful if Bob correctly outputs the $y^{th}$-bit of Alice's input, i.e. when $b=x_y$. The average success probability is \cite{Spekkens2009}
\begin{equation}\label{succprob}
       \mathcal{P}_{\rm avg} = \frac{1}{8} \sum_{y \in \{1,2\}, x \in \{0,1\}^2} p(b = x_y | P_{x},M_y). 
\end{equation}
The key observation is that this RAC task can be decomposed into two distinct, although not independent, binary discrimination tasks. To make this precise, consider the first bit $x_1$. Define the marginal ensembles
\begin{equation}\label{pren1}
    P_0^{(1)}=\frac{1}{2}\qty{P_{00},P_{01}}; \ \  P_1^{(1)}=\frac{1}{2}\qty{P_{10},P_{11}}.
\end{equation}
Operationally, $P_{x_1}^{(1)}$ corresponds to preparing  either $P_{x_1 0}$ or $P_{x_1 1}$ with equal probability. Thus, when Bob receives the input $y=1$, his task is to infer whether the communicated preparation originated from the ensemble $P_0^{(1)}$ or $P_1^{(1)}$. An analogous construction applies to the second bit $x_2$.
\begin{Propos}
   In a 2$\to$1 RAC, the average success probability can be expressed in terms of the distinguishabilities of the marginal ensembles as 
   \begin{equation}
       \mathcal{P}_{\rm avg} \leq \frac{1}{2}+\frac{\Delta_1+\Delta_2}{4},
       \label{avgsuccfmp}
   \end{equation}
where $\Delta_y:=\Delta(P_0^{(y)},P_1^{(y)})$. The bound is saturated when Bob performs the optimal measurements for the corresponding discrimination tasks.
\end{Propos}

\begin{proof}
For each input $y \in \qty{1,2}$, Bob’s task reduces to  discriminating operationally between two marginal ensembles $P_0^{(y)}$ and $P_1^{(y)}$. The optimal success probability for this binary discrimination problem is 
\begin{equation}
    p_g^{\rm opt}\qty(P_0^{(y)},P_1^{(y)})=\frac{1}{2}\qty(1+\Delta_y).
\end{equation}
Averaging over the uniformly distributed inputs $y$ yields the stated bound. Thus, the RAC task reduces to maximising $\Delta_1+\Delta_2$, subject to the constraints imposed by the underlying physical theory. Further details are provided in Appendix~\ref{apasppd}.
\end{proof}

Since both marginal ensembles originate from the same underlying set of preparations, $\Delta_1$ and $\Delta_2$ are not independent, but are jointly constrained by the underlying physical theory. In particular, within the preparation-noncontextual ontological models, the RAC success probability is bounded by $3/4$ \cite{Spekkens2009}, which also coincides with the optimal average success probability when communication between Alice and Bob is restricted to a single classical bit \cite{Ambainis2009}. Substituting this into  Eq.~\eqref{avgsuccfmp} yields the corresponding constraint: $\Delta_1+\Delta_2 \underset{\rm PNC/1 cbit}{\leq} 1$.
By contrast, we show that quantum theory allows larger jointly attainable values of $(\Delta_1,\Delta_2)$ for suitably incompatible measurements. This leads to a violation of the above bound and, hence, to a quantum advantage in the RAC.

%%%%%%%%%%%%%%%%%%%%%%%%%%%%%%%%%%%%%%%%%%%%%%%%%%%%%%%%%%%%%%%%%%%%%%%%%%%%%%%%%%%%%%%%%%%%%%%%%%%%%%%%%%%%%%%%%%%%%%%%%%%%%%%%%%%%%%%%%%%%%%%%%%%%%%%%%%%%%%%%%%%%%%%%%%%%%%%%%%%%%%%%%%%%%%%%%%%%%%%%

\subsection{Quantum bound on preparation distinguishabilities} \label{21RACoptq}

Here, Alice's preparations are represented by density operators $\rho_x \in \mathscr{L}\qty(\mathcal{H}_2)$, and Bob's measurements are dichotomic qubit measurements $\qty{E_{b|y}}$. The average success probability is
\begin{equation}\label{2bitsuccp}
\begin{aligned}
    \mathcal{P}_{\rm avg}^{(Q)} &= \frac{1}{8} \sum_{x\in\{0,1\}^2 } \sum_{y \in \{1,2\}} \Tr[\rho_{x} E_{b=x_y|y}].
\end{aligned}
\end{equation}
The marginal ensembles are related to $\rho_{x_1 x_2}$ as follows
\begin{equation} \label{alicepr}
\begin{aligned}
    \rho_{0}^{(1)} &= \frac{1}{2}\qty(\rho_{00}+\rho_{01}); \ \ \rho_{1}^{(1)} = \frac{1}{2}\qty(\rho_{10}+\rho_{11}); \\
    \rho_{0}^{(2)} &= \frac{1}{2}\qty(\rho_{00}+\rho_{10}); \ \ \rho_{1}^{(2)} = \frac{1}{2}\qty(\rho_{01}+\rho_{11}).
\end{aligned}
\end{equation}
The distinguishability between $\rho_0^{(y)}$ and $\rho_1^{(y)}$ is
\begin{equation}
  \Delta_y:=\Delta(P_0^{(y)}, P_1^{(y)}) =\max_{\qty{E_{b|y}}} \ \frac{1}{2}\sum_b \abs{\Tr[E_{b|y}\qty(\rho_0^{(y)}-\rho_1^{(y)})]}. 
\end{equation}
For any pair of quantum states, the optimal measurement is the Helstrom measurement \cite{Watrous2018book}. Thus, the preparation distinguishability becomes \cite{Watrous2018book}
\begin{equation}\label{pdqtn}
 \Delta_y= \Delta(P_0^{(y)}, P_1^{(y)}) = \frac{1}{2}\norm{\rho_0^{(y)}-\rho_1^{(y)}}_1,
\end{equation}
where $\norm{\cdot}_1=\Tr[\abs{\cdot}]=\Tr[\sqrt{(\cdot)^{\dagger}(\cdot)}]$ denotes the trace norm.

\begin{thm}\label{thpdbq}
    For any set of qubit preparations $\qty{\rho_{x_1x_2}}$, the associated distinguishabilities satisfy
    \begin{equation}
        \Delta_1^2+\Delta_2^2 \leq 1.
        \label{quantumbound}
    \end{equation}
\end{thm}

\begin{proof}
Let $S=\rho_{00}-\rho_{11}$ and $T=\rho_{01}-\rho_{10}$. Then
\begin{equation}
\Delta_1^2+\Delta_2^2=\frac{1}{16}\!\left(\|S{+}T\|_1^2+\|S{-}T\|_1^2\right)
\le \frac{R}{16}\!\left(\|S{+}T\|_2^2+\|S{-}T\|_2^2\right), \nonumber
\end{equation}
where we used $\|X\|_1^2 \le \mathrm{rank}(X)\|X\|_2^2$ \cite{Watrous2018book}, with $\norm{X}_2$ being the Schatten-2 norm, $R=\max\{\mathrm{rank}(S\pm T)\}\le 2$ for qubit preparations. Applying the parallelogram identity for the Schatten-2 norm leads to (\ref{quantumbound}). Further details are provided in Appendix~\ref{apqbprd}.

\end{proof} 

\begin{Corollary}
For any set of qubit preparations $\qty{\rho_{x_1x_2}}$, 
    \begin{equation}
        \max_{\rho_{x_1x_2}} \Delta_1+\Delta_2=\sqrt{2}, \ \ \text{subject to } \Delta_1^2+\Delta_2^2 \leq 1. \label{optdeltam}
    \end{equation}
Consequently, the optimal quantum success probability is
\begin{equation}\label{optquantumavg}
    \qty(\mathcal{P}_{\rm avg}^{(Q)})_{\rm opt} =\frac{1}{2}\qty(1+\frac{1}{\sqrt{2}}).
\end{equation}
\end{Corollary} 

\begin{proof}
    Applying the Cauchy-Schwarz inequality to $\qty(\Delta_1,\Delta_2)$, and $(1,1)$ gives $
\qty(\Delta_1+\Delta_2 )^2 \leq 2\qty(\Delta_1^2+\Delta_2^2 )$,
which implies Eq.~(\ref{optdeltam}). Eq.~(\ref{optquantumavg}) is achieved by substituting (\ref{optdeltam}) into 
\begin{equation}
    \qty(\mathcal{P}_{\rm avg}^{(Q)})_{\rm opt} =\frac{1}{2}+\frac{1}{4} \max_{\rho_{x_1x_2}}\qty(\Delta_1+\Delta_2), \ \text{subject to } \Delta_1^2+\Delta_2^2 \leq 1. \nonumber
\end{equation}

\end{proof}

\begin{Remark}
The set of quantum-achievable distinguishability pairs  satisfying $\Delta_1^2+\Delta_2^2 \leq 1$ forms a quarter unit disc, while classical/preparation-noncontextual models are restricted to the simplex $\Delta_1+\Delta_2 \leq 1$. The bound $\Delta_1+\Delta_2 \leq \sqrt{2}$, which follows from the optimal RAC performance, is not itself fundamental: not every pair $(\Delta_1, \Delta_2)$ satisfying $\Delta_1+\Delta_2 \in [0, \sqrt{2}]$ is quantum realisable. On the other hand, for each value of $\Delta_1+\Delta_2 \in [0, \sqrt{2}]$, there exists at least one pair of $(\Delta_1, \Delta_2)$ that is quantum realisable (see Appendix~\ref{apqbprd}).
\end{Remark}

%%%%%%%%%%%%%%%%%%%%%%%%%%%%%%%%%%%%%%%%%%%%%%%%%%%%%%%%%%%%%%%%%%%%%%%%%%%%%%%%%%%%%%%%%%%%%%%%%%%%%%%%%%%%%%%%%%%%%%%%%%%%%%%%%%%%%%%%%%%%%%%%%%%%%%%%%%%%%%%%%%%%%%%%%%%%%%%%%%%%%%%%%%%%%%%%%%%%%%%%%%%%%%%%%%%%%%%%%%%%%%%%%%%%%%%%%%%%%%%%%%%%%%%%%%%%%%%%%%%%%%%%%%%%%%%%%%%%

\section{Trade-off between preparation distinguishability, measurement incompatibility, and success probability}

Consider the case, where Bob performs dichotomic unbiased qubit POVMs (unsharp measurements) of the form \cite{Busch2016book}
\begin{equation}\label{bobm}
    \mathcal{B}_y\equiv \qty{ E_{b|y} = \frac{\openone + (-1)^b \lambda_y B_y}{2} }; \ \ \lambda_y \in [0,1],
\end{equation}
where $B_y$ is a Hermitian operator with eigenvalues $\pm 1$.

\begin{thm}
For the 2$\to$1 RAC with preparations $\qty{\rho_{x_1x_2}}$ and dichotomic unbiased POVMs of the form in Eq.~\eqref{bobm}, the average success probability satisfies
    \begin{equation}\label{aspbunsp}
        \mathcal{P}^{(Q)}_{\rm avg}\leq \frac{1}{2}+\frac{1}{4}\qty(\lambda_1 \Delta_1+\lambda_2 \Delta_2),
    \end{equation}
    where $\Delta_y$ denotes the preparation distinguishability of the marginal ensembles defined in Eq.~\eqref{pdqtn}. Equality will be achieved when $B_y$ is the Helstrom measurement corresponding to $\Delta_y$. Consequently, a quantum advantage is obtained whenever 
    \begin{equation}\label{qacunspd}
        \lambda_1 \Delta_1+\lambda_2 \Delta_2>1, \ \ \text{subject to } \Delta_1^2 + \Delta_2^2 \leq 1.
    \end{equation}
\end{thm}

\begin{proof} 
It follows from (\ref{2bitsuccp}) that
\begin{equation}\label{2bitsuccp2}
    \mathcal{P}_{\rm avg}^{(Q)} =\frac{1}{2} + \frac{1}{16} \Big\{ \lambda_1 \Tr[\qty(S+T)B_1 ] +  \lambda_2 \Tr[\qty(S-T)B_2] \Big\}.
\end{equation}
Let $X_y:=(S-(-1)^y T)$. Maximising $\Tr[X_y B_y]$ over observables $-\openone \leq B \leq \openone$ yields $\max\Tr[X_y B_y] = \norm{X_y}_1 =4 \Delta_y$. Thus we obtain Eq.~\eqref{aspbunsp}, and hence, the condition for quantum advantage is given by Eq.~\eqref{qacunspd}.
\end{proof}

\begin{Lemma}\label{quantumadvantagelemma}
A quantum advantage $\lambda_1 \Delta_1+\lambda_2 \Delta_2>1$ is achievable if and only if Bob's measurement are incompatible, i.e. $\lambda_1^2+\lambda_2^2 >1$, provided $\qty{B_1,B_2}=0$. 
\end{Lemma}

\begin{proof}
Applying the Cauchy-Schwarz inequality gives
   \begin{equation}
       \lambda_1 \Delta_1+\lambda_2 \Delta_2 \leq \sqrt{\lambda_1^2+\lambda_2^2} \sqrt{\Delta_1^2 + \Delta_2^2}.
   \end{equation}
Since $\Delta_1^2 + \Delta_2^2 \leq 1$, it follows that $\lambda_1 \Delta_1+\lambda_2 \Delta_2 \leq \sqrt{\lambda_1^2+\lambda_2^2}$. Therefore, achieving a quantum advantage requires $\lambda_1^2+\lambda_2^2>1$. For $\qty{B_1,B_2}=0$, this condition is equivalent to measurement incompatibility \cite{Carmeli2012,Heinosaari2016}, thereby establishing necessity.

To prove sufficiency, consider preparations satisfying $\Delta_y=\lambda_y/\sqrt{\lambda^2_1+\lambda^2_2}$. Then $\Delta_1^2 + \Delta_2^2 = 1$, and $\lambda_1 \Delta_1+\lambda_2 \Delta_2 = \sqrt{\lambda_1^2+\lambda_2^2}$. Hence, whenever $\lambda_1^2+\lambda_2^2>1$, a quantum advantage follows immediately.
\end{proof}
We now consider the symmetric scenario in which $\lambda_1=\lambda_2=\lambda^{(s)}$. The threshold value of the unsharpness parameter for quantum advantage is
\begin{equation}
    \lambda^{(s)} > \frac{1}{\Delta_1+\Delta_2} \geq \lambda_c^{(s)} =  \frac{1}{\Delta_1+\sqrt{1-\Delta_1^2}}.
\end{equation}
The minimal threshold corresponds to $\max(\Delta_1+\Delta_2)=\sqrt{2}$, subject to $\Delta_1^2+\Delta_2^2 \leq 1$, yielding $\min \lambda_c^{(s)}=1/\sqrt{2}$. 

We next consider the asymmetric scenario, $\lambda_1=1$ and $\lambda_2=\lambda^{(\rm as)}$. The condition for quantum advantage becomes
\begin{equation} \label{lamsingcr}
  \lambda^{(\rm as)} > \frac{1-\Delta_1}{\Delta_2} \geq \lambda^{(\rm as)}_c = \sqrt{\frac{1-\Delta_1}{1+\Delta_1}}.
\end{equation}
Here, $\lambda^{(\rm as)}_c$ is a monotonically decreasing function of $\Delta_1$. Unlike the symmetric case, the threshold $\lambda^{(\rm as)}_c$ is not minimised at $\Delta_1=\Delta_2=1/\sqrt{2}$. Instead, in the limit $\Delta_1\to 1$ (and hence $\Delta_2 \to 0$), the threshold approaches zero, $\lambda^{(\rm as)}_c \to 0$. This shows that even when the preparation distinguishabilities approach the classical boundary, $\Delta_1 + \Delta_2 \to 1$, nearly compatible measurements can still yield a quantum advantage, thereby opening a new possibility for sequential RACs. The details are presented in Fig.~\ref{figsasyvarto} and Appendix \ref{regionunsharpfigures}.

\begin{figure}[t!]
\includegraphics[width=\linewidth]{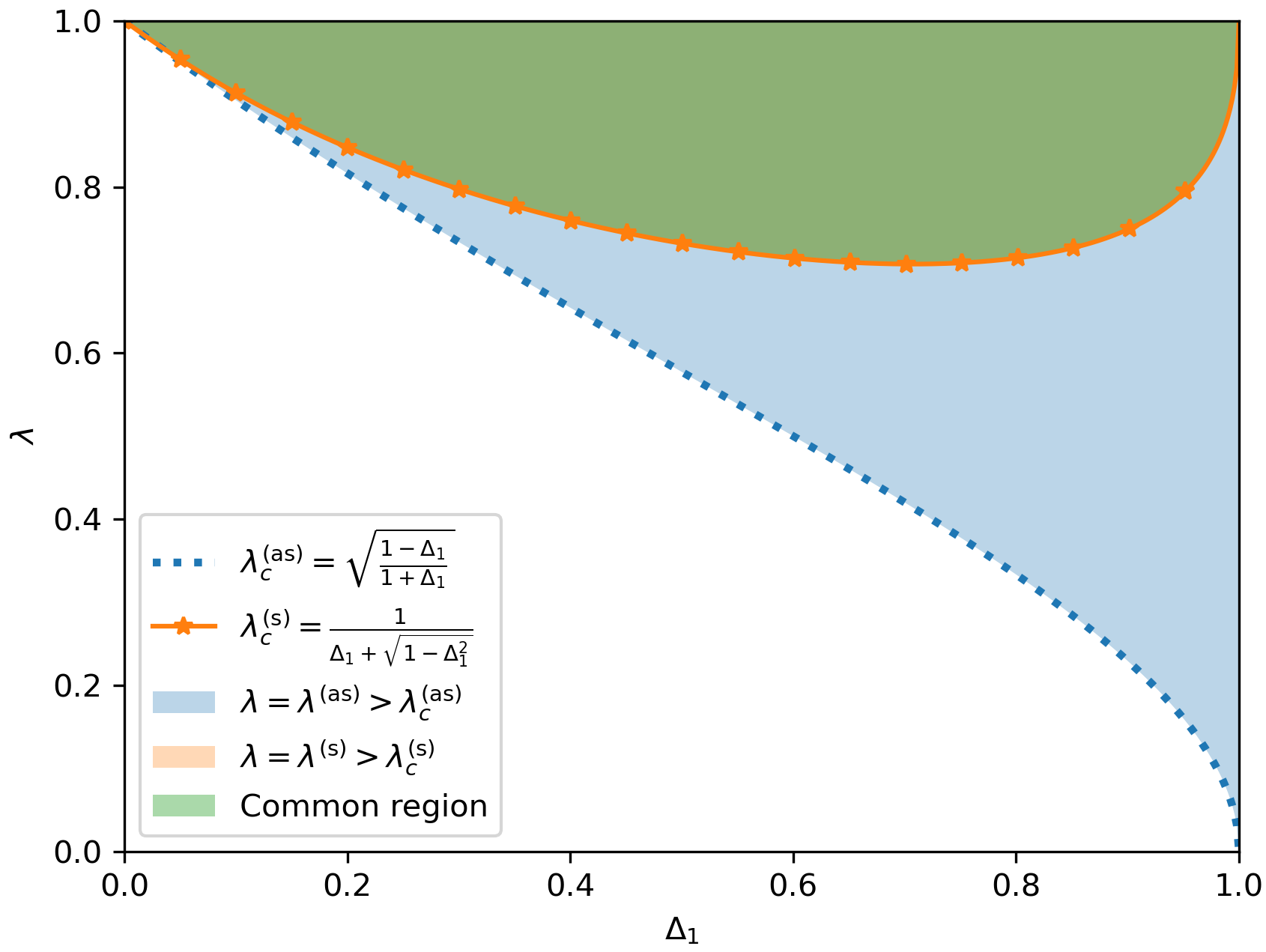}
\caption{Comparison of the symmetric and asymmetric thresholds of the unsharpness parameter as functions of $\Delta_1 \leq \sqrt{1-\Delta_2^2}$. The shaded regions indicate where a quantum advantage is obtained.}
\label{figsasyvarto}
\end{figure}

%%%%%%%%%%%%%%%%%%%%%%%%%%%%%%%%%%%%%%%%%%%%%%%%%%%%%%%%%%%%%%%%%%%%%%%%%%%%%%%%%%%%%%%%%%%%%%%%%%%%%%%%%%%%%%%%%%%%%%%%%%%%%%%%%%%%%%%%%%%%%%%%%%%%%%%%%%%%%%%%%%%%%%%%%%%%%%%%%%%%%%%%%%%%%%%%%%%%

\section{Sequential RAC and quantum advantage for unbounded number of receivers}

Consider the sequential RAC scenario illustrated in Fig.~\ref{seq}. Each Bob$^{(k)}$ performs one of the following two measurements: a projective measurement of $B_1$, with projectors $B_{\pm 1|1}^{(k)} = \frac{1}{2}\qty(\openone \pm B_1)$, and an unsharp measurement of $B_2$, with POVM elements $E^{(k)}_{\pm} = \frac{1}{2}\qty(\openone\pm \lambda_k B_2)$, with unsharpness parameter $\lambda_k \in \qty[0,1]$. Considering the same  $B_1$, $B_2$ for all Bobs serves our purpose while simplifying the calculations.

Let $\rho_{x_1x_2}^{(k)}$ denote the state received by Bob$^{(k)}$. These states evolve recursively according to  $\rho_{x_1x_2}^{(k)}=\Phi^{(k-1)}(\rho_{x_1x_2}^{(k-1)})$, where $\Phi^{(k-1)}$ is the quantum channel induced by the measurements of Bob$^{(k-1)}$. Since measurement outcomes are not communicated to subsequent receivers, $\Phi^{(k)}$ corresponds to the non-selective implementation of the measurements by the preceding Bob
\begin{equation}
   \Phi^{(k)}\qty(\rho_{x_1x_2}^{k-1})= \frac{1}{2}\sum_{b \in \qty{\pm}}\qty(B_{b|1} \rho^{k-1}_{x_1x_2} B_{b|1}+ K^{(k)}_{b} \rho^{k-1}_{x_1x_2} K^{(k)}_{b}),
\end{equation}
where $K^{(k)}_{b} = \sqrt{E^k_{b}}$ are the Kraus operators corresponding to the Luder's instrument associated with the unsharp POVM \cite{Busch2016book}, and depend on $\lambda_k$ (see Appendix~\ref{apcapd} for explicit expressions). 

 Upon receiving the input $y_k \in \{1,2\}$, Bob$^{(k)}$ seeks to discriminate between the two marginal ensembles $\qty{\rho_{0}^{(y,k)},\rho_{1}^{(y,k)}}$, with $y=y_k$, where $\rho_{a}^{(y,k)}$ for all $a \in \{0,1\}$ and $y \in \{1,2\}$, is defined by Eq.(\ref{alicepr}) with $\rho_{x_1x_2}$  replaced by  $\rho^{(k)}_{x_1x_2}$. 

\begin{Lemma}
Let $B_{1}$ and $B_2$ be the Helstrom measurements used by Bob$^{(1)}$ for discriminating the pairs $\qty{\rho_{0}^{(y,1)},\rho_{1}^{(y,1)}}$ for $y=1$ and $y=2$, respectively,  and suppose that $\qty{B_1,B_2}=0$. Then the same observables $B_{1}$ and $B_2$ remain the Helstrom measurements for discriminating the pairs $\qty{\rho_{0}^{(y,k)},\rho_{1}^{(y,k)}}$ for $y=1$ and $y=2$, respectively, for all $k$, where 
 \begin{equation}
   \rho_{x}^{(y,k)}=\Phi^{(k-1)}\qty(\rho_{x}^{(y,k-1)}), \ \ \forall x \in \qty{0,1}, \ y\in\qty{1,2}.
 \end{equation}
Consequently, the preparation distinguishability can be expressed as
 \begin{equation} \label{pdktnpre}
     \Delta_{y}^{(k)}=\frac{1}{2}\norm{\rho_{0}^{(y,k)} -\rho_{1}^{(y,k)}}_1  \ \ \forall  y \in \{1,2\}.
 \end{equation}
\end{Lemma}

\begin{proof}
Using the action of the channel $\Phi^{(k-1)}$, together with $\qty{B_1,B_2}=0$, one obtains
\begin{equation}
 \Phi^{(k-1)}(B_1) =\frac{1}{2}\qty(1+\sqrt{1-\lambda_{k-1}^2}) B_1, \ \ \Phi^{(k-1)}(B_2) =\frac{1}{2}B_2.
\end{equation}
 Hence, if $\rho_{0}^{(y,k-1)} -\rho_{1}^{(y,k-1)} \propto B_y$, then $\rho_{0}^{(y,k)} -\rho_{1}^{(y,k)} = \Phi^{(k-1)}\qty(\rho_{0}^{(y,k-1)} -\rho_{1}^{(y,k-1)}) \propto \Phi^{(k-1)}\qty(B_y) \propto B_y$. Assuming $\rho_{0}^{(y,1)} -\rho_{1}^{(y,1)} \propto B_y$, the result follows by induction. See Appendix~\ref{apcapd} for its details and the proof of the following corollary. 
\end{proof}

\begin{Corollary}
    For the above measurement scheme, and for anticommuting observables $\qty{B_1,B_2}=0$, the preparation distinguishabilities evolve under the action of Bob$^{(k-1)}$ as
    \begin{equation}\label{deltainductmain}
  \Delta_1^{(k)}=\frac{1}{2}\qty(1+\sqrt{1-\lambda_{k-1}^2})\Delta_1^{(k-1)}, \quad      \Delta_2^{(k)}=\frac{1}{2}\Delta_2^{(k-1)}. 
    \end{equation}
\end{Corollary}

The corresponding average success probability is therefore
\begin{equation}
\label{unsharpsucprob}
  \mathcal{P}^{(k)}_{\rm avg}=\frac{1}{2}+\frac{1}{4}\qty(\Delta_1^{(k)}+\lambda_k \Delta_2^{(k)}),  
\end{equation}
and depends both on Alice’s preparation strategy and on the sequence of unsharpness parameters $\qty{\lambda_l}_{l=1}^k$. 

\begin{thm}
Consider the sequential 2$\to$1 RAC protocol with dichotomic qubit observables satisfying $\qty{B_1,B_2}=0$, and initial preparation distinguishabilities 
\begin{equation} \label{distpar}
\qty(\Delta_1^{(1)},\Delta_2^{(1)})=\qty(\cos\omega,r\sin\omega), \quad r \leq 1, \ \ \omega\in (0, \pi/2)
\end{equation}
Then, for any integer $N\in \mathbb{Z}^+$, there exists a sequence of unsharpness parameters 
\begin{equation}
0<\lambda_1<\lambda_2<\cdots<\lambda_k<\cdots < \lambda_N<1, 
\label{lambdamonotonic}
\end{equation}
such that all $N$ independent sequential receivers achieve quantum advantages.
\end{thm}

\begin{proof}
A sufficient condition for guaranteeing quantum advantage for all Bob$^{(k)}$, with $k=\in\qty{1,2,\ldots,N}$, follows from
Eqs.~\eqref{deltainductmain} and \eqref{lamsingcr}
    \begin{equation}\label{lambklb}
\lambda_k > \frac{2^{k-1} -\cos\omega \prod_{l=1}^{k-1}\qty(1+\sqrt{1-\lambda_{l}^2})}{r\sin\omega} \ \ \forall k.
    \end{equation}
We consider a sequence $\{\lambda_k\}_{k=1}^N$ satisfying Eq.~(\ref{lambklb}). For any $r\in(0,1]$, we show that the interval $\omega\in(0,\pi/2)$ ensures  $\qty{0<\lambda_1<\cdots<\lambda_k}_{k=1}^N$ if $\lambda_k<1 \ \forall k$.  Subsequently, using the small-angle approximation, we show that $\lambda_k = \order{\omega}$. Therefore, $\omega\to 0^+$ guarantees $\lambda_k<1 \ \forall k$. See Appendix~\ref{apsharecal} for details.
\end{proof}

\begin{Remark}
 Note that when $\omega \to 0^+$, we have $\lambda_k \to 0$ for all $k$ satisfying Eq.(\ref{lambdamonotonic}). Consequently, from Eqs.(\ref{deltainductmain}) and (\ref{distpar}), for all $k$, $\qty(\Delta_1^{(k)},\Delta_2^{(k)}) \to  \qty(1,0)$. Hence, in this limit, the decoding measurements almost completely preserve the preparation distinguishabilities, thereby enabling an arbitrarily long sequence of receivers to each obtain a quantum advantage. This is consistent with the intuition  discussed below Eq.(\ref{lamsingcr}).
\end{Remark}

%%%%%%%%%%%%%%%%%%%%%%%%%%%%%%%%%%%%%%%%%%%%%%%%%%%%%%%%%%%%%%%%%%%%%%%%%%%%%%%%%%%%%%%%%%%%%%%%%%%%%%%%%%%%%%%%%%%%%%%%%%%%%%%%%%%%%%%%%%%%%%%%%%%%%%%%%%%%%%%%%%%%%%%%%%%%%%%%%%%%%%%%%%%%%%%%%%%%%%%%%%%%%%%%%%%%%%%%%%%%%%%%%%%%%%%%%%%%%%%%%%%%%%%%%%%%%%%%%%%%%%%%%%%%%%%%%%%%

\section{Outlook}

We have shown that a single qubit can exhibit quantum advantage in a sequential RAC scenario comprising a single sender and an arbitrarily long sequence of independent receivers. This reveals a previously unexplored operational feature of a qubit: its ability to retain encoded information even after repetitive decoding attempts with quantum advantages, despite the inherently invasive nature of the decoding measurements. Our results, therefore, bring out a non-trivial information-disturbance trade-off in sequential measurement scenarios.

To establish this result, we have identified preparation distinguishability as one of the operational resources governing RAC performance. The notion introduced here is closely related to the concept of $p$-distinguishability introduced in the framework of bounded ontological distinctness \cite{Chaturvedi2020}. Unlike that framework, however, our approach provides an operational characterisation within a communication task. This perspective also suggests a possible route towards a more complete resource-theoretic understanding of RACs, incorporating not only measurement incompatibility \cite{Saha2023,Saha2020,Carmeli2020}, but also preparation distinguishability.

Recent studies have explored the recycling of spatial quantum correlations among sequential observers \cite{Cai2025}. By contrast, the present work establishes unbounded sequential sharing in a fundamentally different setting, viz. a qubit-based prepare-measure communication task. In particular, our result goes beyond the phenomenon of unbounded sequential sharing of nonlocal quantum correlations \cite{Brown2020,Sasmal2024} to a communication framework that is operationally distinct and relies on different underlying resources. More broadly, it opens up the possibility of one-to-many communication protocols not only across space, but also across time.

Several directions for future investigation naturally emerge from our results. In particular, it would be worthwhile to explore the role of preparation distinguishability in more general prepare-measure-based communication scenarios \cite{Gallego2010,Gois2021,Manna2024,Singh2025,Brask2026}. In the present work, for $k$ sequential receivers, we have identified a  family of state preparations characterised by $(\Delta_1,\Delta_2)=(\cos\omega,r \sin\omega)$, where both $\omega$ and $r$ are independent of $k$, with $\omega \to 0^+$, for which every receiver can achieve a quantum advantage even when $k$ becomes arbitrarily large.  It would be interesting to investigate other families of preparations of the form $(\Delta_1,\Delta_2)=(r'(k)\cos\theta(k),r'(k) \sin\theta(k))$ with $r'(k) \leq 1$, which may likewise enable quantum advantage for arbitrarily many sequential receivers. 

%%%%%%%%%%%%%%%%%%%%%%%%%%%%%%%%%%%%%%%%%%%%%%%%%%%%%%%%%%%%%%%%%%%%%%%%%%%%%%%%%%%%%%%%%%%%%%%%%%%%%%%%%%%%%%%%%%%%%%%%%%%%%%%%%%%%%%%%%%%%%%%%%%%%%%%%%%%%%%%%%%%%%%%%%%%%%%%%%%%%%%%%%%%%%%%%%%%%%%%%%%%%%%%%%%%%%%%%%%%%

\section*{Acknowledgements}
SS acknowledges the support from the National Natural Science Fund of China (Grant No. W2533013). SK acknowledges the support from Digital Horizon Europe project, FoQaCia (\textit{Foundations of quantum computational advantage}), GA no 202070558, funded by the European Union and NSERC (Canada).

%%%%%%%%%%%%%%%%%%%%%%%%%%%%%%%%%%%%%%%%%%%%%%%%%%%%%%%%%%%%%%%%%%%%%%%%%%%%%%%%%%%%%%%%%%%%%%%%%%%%%%%%%%%%%%%%%%%%%%%%%%%%%%%%%%%%%%%%%%%%%%%%%%%%%%%%%%%%%%%%%%%%%%%%%%%%%%%%%%%%%%%%%%%%%%%%%%%%%%%%%%%%%%%%%%%%%%%%%%%%%%%%%%%%%%%%%%%%%%%%%%%%%%%%%%%%%%%%%%%%%%%%%%%%%%%%%%%%%%%%%%%%%%%%

\bibliography{references} 

%%%%%%%%%%%%%%%%%%%%%%%%%%%%%%%%%%%%%%%%%%%%%%%%%%%%%%%%%%%%%%%%%%%%%%%%%%%%%%%%%%%%%%%%%%%%%%%%%%%%%%%%%%%%%%%%%%%%%%%%%%%%%%%%%%%%%%%%%%%%%%%%%%%%%%%%%%%%%%%%%%%%%%%%%%%%%%%%%%%%%%%%%%%%%%%%%%%%%%%%%%%%%%%%%%%%%%%%%%%%%%%%%%%%%%%%%%%%%%%%%%%%%%%%%%%%%%%%%%%%%%%%%%%%%%%%%%%%%%%%%%%%%%%%

\appendix
\onecolumngrid

\section{Relating   Guessing Probability and Preparation Distinguishability} \label{apcgppd}

Consider a binary discrimination task between the two preparations $P_0$ and $P_1$ chosen equal prior probabilities. Suppose, under a fixed measurement $M$ with outcomes $m\in\qty{0,1}$, $P_0$ is guessed upon observing $m=0$, and  $P_1$ is guessed otherwise. Hence, with this fixed measurement, the average success probability is 
\begin{equation}\label{fixmgdistas}
    p_{\rm g}(P_0,P_1|M)=\frac{1}{2}\qty[p(0|P_0,M)+p(1|P_1,M)].
\end{equation}
Optimising over all measurements, we get
\begin{align}
    p_{\rm g}^{\rm opt}(P_0,P_1) = \frac{1}{2}  \max_M \sum_{m\in \{0,1\}} p(m|P_m,M) \label{2statedistp001}
\end{align}

\begin{Lemma}\label{cgppd}
For any pair of preparations $P_0,P_1$,    
\begin{equation}\label{cgppdeq}
p_{\rm g}^{ \rm opt}(P_0,P_1)=\frac{1}{2}\qty[1+\Delta(P_0,P_1)].
\end{equation}
\end{Lemma}

\begin{proof}
Rewriting Eq.(\ref{2statedistp001}), it follows that
\begin{align}
    p_{\rm g}^{\rm opt}(P_0,P_1) =\frac{1}{2}+\frac{1}{2}\max_M \qty{p(0|P_0,M)-p(0|P_1,M)}. \label{2statedistp00}
\end{align}
The term in brackets quantifies how strongly outcome `$0$' of the optimal measurement favours the preparation $P_0$ over $P_1$.  The optimisation over all measurement strategies ensures that, for the optimal measurement $M*$, this bias is always non-negative., i.e. $p(0|P_0,M*)>p(0|P_1,M*)$. Consequently, Bob’s optimal strategy is to guess $P_0$ whenever he observes outcome `$0$', and to guess $P_1$ otherwise. This choice guarantees that the difference term in Eq. (\ref{2statedistp00}) contributes positively to the success probability, so that the maximisation can be written as
\begin{equation}\label{2statedistp1}
    p_{\rm g}^{\rm opt}(P_0,P_1)=\frac{1}{2}+\frac{1}{2}\max_M \abs{p(0|P_0,M)-p(0|P_1,M)} .
\end{equation}
However, this expression implicitly assumes that outcome `$0$' is the one that favours $P_0$. In general, the outcome that provides the higher bias could be either `$0$' or `$1$', depending on the measurement and the preparation. For dichotomic measurements, the outcome probabilities are related by $p(1|P,M)=1-p(0|P,M)$. Therefore, the absolute difference for one outcome equals that for the other, and we can express it symmetrically as
\begin{equation}
    \abs{p(0|P_0,M)-p(0|P_1,M)}=\abs{p(1|P_0,M)-p(1|P_1,M)} =\frac{1}{2} \sum_m \abs{p(m|P_0,M)-p(m|P_1,M)}.
\end{equation}
Substituting this into the previous expression gives
\begin{equation}\label{2statedistp1f}
    p_{\rm g}^{\rm opt}(P_0,P_1)=\frac{1}{2}+\frac{1}{4}\max_M \sum_m \abs{p(m|P_0,M)-p(m|P_1,M)}=\frac{1}{2}\qty[1+\Delta(P_0,P_1)],
\end{equation}
which follows from the definition of $\Delta(P_0,P_1)$.
\end{proof}

%%%%%%%%%%%%%%%%%%%%%%%%%%%%%%%%%%%%%%%%%%%%%%%%%%%%%%%%%%%%%%%%%%%%%%%%%%%%%%%%%%%%%%%%%%%%%%%%%%%%%%%%%%%%%%%%%%%%%%%%%%%%%%%%%%%%%%%%%%%%%%%%%%%%%%%%%%%%%%%%%%%%%%%%%%%%%%%%%%%%%%%%%%%%%%%%%%%%%%%%%%%%%%%%%%%%%%%%%%%%%%%%%%%%%%%%%%%%%%%%%%%%%%%%%%%%%%%%%%%%%%%%%%%%%%%%%%%%%%%%%%%%%%%%

\section{Expressing the Average Success Probability of  $2 \to 1$ RAC in terms of Preparation Distinguishabilities} \label{apasppd}

In the $2 \to 1$ RAC, the marginal ensembles associated with the first bit $x_1$ received by Alice is defined as
\begin{equation}\label{pren1}
    P_0^{(1)}=\frac{1}{2}\qty{P_{00},P_{01}}; \ \  P_1^{(1)}=\frac{1}{2}\qty{P_{10},P_{11}}.
\end{equation}
Operationally, $P_{x_1}^{(1)}$ corresponds to preparing $P_{x_1 0}$ or $P_{x_1 1}$ with equal probability. Consequently, for any measurement $M_1$,
\begin{equation}\label{ensdesopx10}
p\qty(m|P_{x_1}^{(1)},M_1)=\frac{1}{2}\qty[p(m|P_{x_1 0},M_1)+p(m|P_{x_1 1},M_1)].
\end{equation}
An analogous construction applies for the second bit $x_2$, for which we define
\begin{equation}\label{pren2}
P_0^{(2)}=\frac{1}{2}\qty{P_{00},P_{10}}; \ \  P_1^{(2)}=\frac{1}{2}\qty{P_{01}, P_{11}};
\end{equation}
with the corresponding operational relation for any measurement $M_2$:
\begin{equation}\label{ensdesopx11}
p\qty(m|P_{x_2}^{(2)},M_2)=\frac{1}{2}\qty[p(m|P_{0x_2},M_2)+p(m|P_{1 x_2},M_2)].
\end{equation}
For each input $y$ received by Bob, Bob's task of recovering $x_y$ is operationally equivalent to a binary discrimination problem between the two ensembles $P_0^{(y)}$ or $P_1^{(y)}$. His optimal success probability for this task is, therefore, given by the optimal guessing probability
\begin{equation}\label{guesoppr}
    p^{\rm opt}_g\qty(P_0^{(y)},P_1^{(y)})=\frac{1+\Delta_y}{2}; \ \ \Delta_y:=\Delta\qty(P_0^{(y)},P_1^{(y)}),
\end{equation}
where $\Delta_y$ quantifies the preparation distinguishability of the marginal ensembles associated with the bit $x_y$.

In the RAC task, for a fixed value of $y$, Bob aims to correctly guess the bit $x_y$ encoded in Alice’s preparation, by performing the measurement $M_y$ with outcomes being denoted by $b$. Since Alice’s input $x=(x_1,x_2)$ is uniformly distributed, each preparation $P_x$ occurs with probability $\frac{1}{4}$. The corresponding success probability is therefore
\begin{equation}\label{racsuccprep1}
p_{\rm succ}(x_y)= \frac{1}{4} \sum_x p(b=x_y|P_x,M_y).
\end{equation} 
We now show that this quantity coincides with the success probability of discriminating between the ensembles $P_0^{(y)}$ and $P_1^{(y)}$ using the measurement $M_y$.  At first, consider the case $y=1$. From Eq.~(\ref{racsuccprep1}), we obtain
    \begin{equation}\label{consucdisy1}
\begin{aligned}
p_{\rm succ}(x_{y=1}) &= \frac{1}{4} \sum_{x_1,x_2} p\qty(b=x_1 \mid P_{x_1x_2},M_1) \\
& =  \frac{1}{4} \sum_{x_1}\qty[ p\qty(b=x_1 \mid P_{x_10},M_1) +p\qty(b=x_1 \mid P_{x_1 1},M_1)] \\
    & = \frac{1}{4} \qty[p\qty(b=0 \mid P_{00},M_1)+p\qty(b=0 \mid P_{01},M_1)+p\qty(b=1 \mid P_{10},M_1)+p\qty(b=1 \mid P_{11},M_1)] \\
    &= \frac{1}{2}\qty[\frac{1}{2}\qty{p\qty(b=0 \mid P_{00},M_1)+p\qty(b=0 \mid P_{01},M_1)}+\frac{1}{2}\qty{p\qty(b=1 \mid P_{10},M_1)+p\qty(b=1 \mid P_{11},M_1)}] \\
    &=\frac{1}{2}\qty[p\qty(0 \mid P_{0}^{(1)},M_1)+p\qty(1 \mid P_{1}^{(1)},M_1)] \quad \text{[From Eq.~(\ref{ensdesopx10})]} \\
    &= \frac{1}{2} \sum_m p\qty(m \mid P_m^{(1)},M_1) \\
    &= p_{\rm g}\qty(P_0^{(1)},P_1^{(1)}|M_1). \quad \text{[From Eq.~(\ref{fixmgdistas})]}
    \end{aligned}
\end{equation}
Thus, it follows from Eqs.~(\ref{guesoppr}) and (\ref{consucdisy1}) that
\begin{equation} \label{sucdisy1f}
    p_{\rm succ}(x_{y=1})=p_{\rm g}\qty(P_0^{(1)},P_1^{(1)}|M_1) \leq \frac{1}{2}\qty(1+\Delta_1).
\end{equation}
with equality when $M_1=M_1^*$ is chosen for optimally  discriminating between $P_0^{(1)}$ and $P_1^{(1)}$. An analogous argument applies for $y=2$, yielding
\begin{equation}\label{sucdisy2f}
p_{\rm succ}(x_{y=2}) = p_{\rm g}\qty(P^{(2)}_{0},P^{(2)}_1|M_2) \leq \frac{1}{2}\qty(1+\Delta_2).
\end{equation}
Averaging over both inputs $y\in\qty{1,2}$, from Eqs.~(\ref{sucdisy1f}) and (\ref{sucdisy2f}), the \textit{average} RAC success probability becomes
\begin{equation}\label{avgsuccf}
        \mathcal{P}_{\rm avg}=\frac{1}{2}\qty[p_{\rm succ}(x_{y=1})+p_{\rm succ}(x_{y=2})] \leq \frac{1}{2}+\frac{\Delta_1+\Delta_2}{4},
\end{equation}
with the bound being tight whenever Bob implements the optimal measurement $(M_y^*)$ for each marginal ensemble description.

%%%%%%%%%%%%%%%%%%%%%%%%%%%%%%%%%%%%%%%%%%%%%%%%%%%%%%%%%%%%%%%%%%%%%%%%%%%%%%%%%%%%%%%%%%%%%%%%%%%%%%%%%%%%%%%%%%%%%%%%%%%%%%%%%%%%%%%%%%%%%%%%%%%%%%%%%%%%%%%%%%%%%%%%%%%%%%%%%%%%%%%%%%%%%%%%%%%%%%%%%%%%%%%%%%%%%%%%%%%%%%%%%%%%%%%%%%%%%%%%%%%%%%%%%%%%%%%%%%%%%%%%%%%%%%%%%%%%%%%%%%%%%%%%

\section{Derivation of $\Delta_1^2 + \Delta_2^2 \leq 1$ for Qubit Preparations.}\label{apqbprd}

We now derive an quantum upper bound on the quantities $\Delta_1$ and $\Delta_2$. Define the Hermitian operators
\begin{equation}\label{sandt}
    S:=\rho_{00}-\rho_{11}; \quad  T:= \rho_{01}-\rho_{10}.
\end{equation}
From the definition of the marginal ensembles, the preparation distinguishabilities are given by
\begin{equation}\label{qb1}
\begin{aligned}
    \Delta_1&:=\frac{1}{2} \norm{\rho_{0}^{(1)}-\rho_{1}^{(1)}}_1=\frac{1}{4}\norm{S+T}_1, \\  \Delta_2&:=\frac{1}{2} \norm{\rho_{0}^{(2)}-\rho_{1}^{(2)}}_1=\frac{1}{4}\norm{S-T}_1.
    \end{aligned} \implies  \Delta_1^2+\Delta_2^2=\frac{1}{16}\qty(\norm{S+T}^2_1+\norm{S-T}^2_1).
\end{equation}
The trace norm appearing above generally requires diagonalising the operator $\rho_{0}^{(y)}-\rho_{1}^{(y)}$, which may become computationally demanding. It is, therefore, useful to relate it to the Schatten-2 norm $\norm{\cdot}_2$, which is also related to the  Hilbert-Schmid (HS) distance $D_{\rm HS}(\cdot)$ as follows
\begin{equation}\label{HSdistdef}
    \norm{\rho_0-\rho_1}^2_2 = D_{\rm HS}(\rho_0,\rho_1):=\Tr[(\rho_0-\rho_1)^2].
\end{equation}
Although the HS distance does not possess the same direct operational interpretation as the trace distance \cite{Ozawa2000}, it admits a simpler algebraic form. The Schatten norm inequalities imply \cite{Watrous2018book}
\begin{equation}
    \norm{X}_2 \leq \norm{X}_1 \leq \sqrt{\rank(X)} \norm{X}_2, \label{normrelations}
\end{equation}
where  the first inequality follows from the monotonicity of the Schatten norms and the second one follows from the Cauchy-Schwartz inequality \cite{Coles2019}. From Eqs.(\ref{qb1}) and (\ref{normrelations}), we get
\begin{align}
    \Delta_1^2+\Delta_2^2 \leq \frac{1}{16} (R_+ \norm{S+T}^2_2+ R_-\norm{S-T}^2_2) \leq \frac{R}{16} (\norm{S+T}^2_2+ \norm{S-T}^2_2),
\end{align}
where $R_{\pm} = \rank(S \pm T)$ and $R:=\max(R_+,R_-)$. Since $S \pm T$ are operators on two-dimensional Hilbert spaces, $R \leq 2$. Using the parallelogram identity,
\begin{equation}
    \norm{S+T}^2_2+\norm{S-T}^2_2 = 2 \qty(\norm{S}^2_2+\norm{T}^2_2) \implies \Delta_1^2 +\Delta_2^2 \leq \frac{R}{8} \qty(\norm{S}^2_2+\norm{T}^2_2).
\end{equation}
\begin{figure}[t!]
\includegraphics[width=0.5\linewidth]{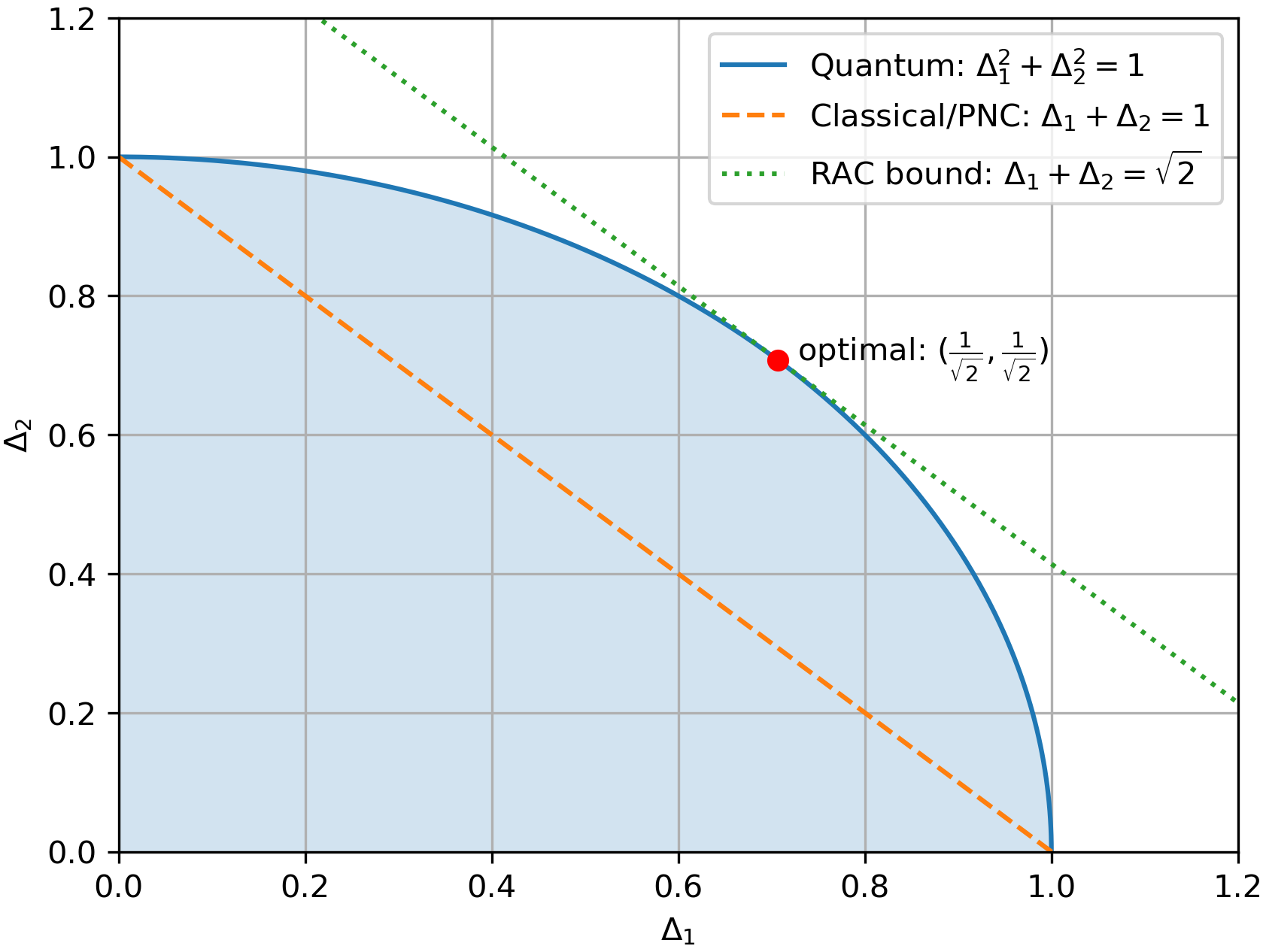}
\caption{Comparison of constraints on preparation distinguishabilities $(\Delta_1,\Delta_2)$ in the 2$\to$1 RAC. The set achievable with any classical/preparation-noncontextual models are restricted to the simplex $\Delta_1+\Delta_2 \leq 1$, whereas the set of quantum-achievable distinguishability pairs forms a quarter unit disc $\Delta_1^2+\Delta_2^2 \leq 1$, with the optimal quantum point at $\qty(\frac{1}{\sqrt{2}},\frac{1}{\sqrt{2}})$.}
\label{figpdbil}
\end{figure}
Since $S$ and $T$ are Hermitian and are differences of density operators, we have
\begin{equation}
\begin{aligned}
    \norm{S}_2^2&=\Tr[(\rho_{00}-\rho_{11})^2]=\Tr[\rho_{00}^2]+\Tr[\rho_{11}^2]-2\Tr[\rho_{00}\rho_{11}],\\
    \norm{T}_2^2&=\Tr[(\rho_{01}-\rho_{10})^2]=\Tr[\rho_{01}^2]+\Tr[\rho_{10}^2]-2\Tr[\rho_{01}\rho_{10}].
\end{aligned}
\end{equation}
For any density operator $\rho$, $\Tr[\rho^2]\leq 1$, with equality only for pure states. Also, for any two density operators $\rho,\sigma$, the term $\Tr[\rho\sigma]\geq 0$. Therefore,
\begin{equation}
    \norm{S}^2_2 \leq \Tr[\rho_{00}^2]+\Tr[\rho_{11}^2] \leq 2, \ \ \norm{T}^2_2 \leq \Tr[\rho_{01}^2]+\Tr[\rho_{10}^2] \leq 2.
\end{equation}
Hence,
\begin{equation}
\Delta_1^2+\Delta_2^2 \leq \frac{R}{2} \leq 1. 
\end{equation} 
In the above, for the equality to be attained, $\rho_{00}$ and $\rho_{11}$ must be mutually orthogonal pure states, and $\rho_{01}$ and $\rho_{10}$ must likewise be mutually orthogonal pure states.

Using the Cauchy-Schwarz inequality, it can be shown that for any set of qubit preparations $\qty{\rho_{x_1x_2}}$, 
    \begin{equation}
        \max_{\rho_{x_1x_2}} \Delta_1+\Delta_2=\sqrt{2}, \ \ \text{subject to } \Delta_1^2+\Delta_2^2 \leq 1. \label{optdelta}
    \end{equation}
The constraint $\Delta_1^2+\Delta_2^2 \leq 1$ admits a simple geometric interpretation in the $(\Delta_1,\Delta_2)$ plane as shown in Fig.~\ref{figpdbil}. The quantum maximum of $\Delta_1+\Delta_2$ is attained for $\Delta_1 = \Delta_2 = 1/\sqrt{2}$. Note that $\Delta_1=\Delta_2=1/\sqrt{2}$ implies that Alice's four pure qubit states lie on the vertices of a square in the equatorial plane of the Bloch sphere.

%%%%%%%%%%%%%%%%%%%%%%%%%%%%%%%%%%%%%%%%%%%%%%%%%%%%%%%%%%%%%%%%%%%%%%%%%%%%%%%%%%%%%%%%%%%%%%%%%%%%%%%%%%%%%%%%%%%%%%%%%%%%%%%%%%%%%%%%%%%%%%%%%%%%%%%%%%%%%%%%%%%%%%%%%%%%%%%%%%%%%%%%%%%%%%%%%%%%%%%%%%%%%%%%%%%%%%%%%%%%%%%%%%%%%%%%%%%%%%%%%%%%%%%%%%%%%%%%%%%%%%%%%%%%%%%%%%%%%%%%%%%%%%%%%%%%%%%%%%%%%%%%%%%%%%%%%%%%%%%%%%%

\section{Regions of the Unsharpness Parameter that enable a Quantum Advantage in the 2$\to$1 RAC} \label{regionunsharpfigures}

In the 2$\to$1 RAC,   consider that Bob performs dichotomic unbiased qubit POVMs (unsharp measurements) of the form: $\mathcal{B}_y\equiv \qty{ E_{b|y} = \frac{\openone + (-1)^b \lambda_y B_y}{2} }$ with $\lambda_y \in [0,1]$,
where $B_y$ is a Hermitian operator with eigenvalues $\pm 1$. A quantum advantage is achieved whenever 
    \begin{equation}\label{qacunspdappendx}
        \lambda_1 \Delta_1+\lambda_2 \Delta_2>1, \ \ \text{subject to } \Delta_1^2 + \Delta_2^2 \leq 1,
    \end{equation}
 where $\Delta_y$ denotes the preparation distinguishability of the marginal ensembles. 

In the symmetric case, $\lambda_1=\lambda_2=\lambda^{(s)}$. In this case, the condition for a quantum advantage is given by
\begin{equation}
    \lambda^{(s)} >  \lambda^{(s)}_c=\frac{1}{\Delta_1+\Delta_2}, \ \ \text{subject to } \Delta_1^2+\Delta_2^2 \leq 1.
\end{equation}
The minimum of $\lambda^{(s)}_c$ is attained at $\Delta_1=\Delta_2=1/\sqrt{2}$, yielding $\lambda_c^{(s)}=1/\sqrt{2}$. The corresponding trade-off between preparation distinguishability and measurement incompatibility is illustrated in Fig.~\ref{figlameq}.

In the asymmetric case, setting $\lambda_1=1$ and $\lambda_2=\lambda^{(\mathrm{as})}$, the condition becomes
\begin{equation}
\lambda^{(\mathrm{as})} > \lambda^{(\mathrm{as})}_c= \frac{1-\Delta_1}{\Delta_2}, \ \ \text{subject to } \Delta_1^2+\Delta_2^2 \leq 1.
\end{equation}
In contrast to the symmetric scenario, the minimum of $\lambda^{(\mathrm{as})}_c$ is not achieved at $\Delta_1=\Delta_2=1/\sqrt{2}$ (where $\lambda_c^{(\mathrm{as})}=\sqrt{2}-1$), but instead $\lambda_c^{(\mathrm{as})}\to 0$ as $\Delta_1\to 1$ (and $\Delta_2\to 0$). The corresponding trade-off is shown in Fig.~\ref{figlamsing}.

\begin{figure*}[t!]
    \centering
    \begin{subfigure}{0.48\textwidth}
        \centering
        \includegraphics[width=\linewidth]{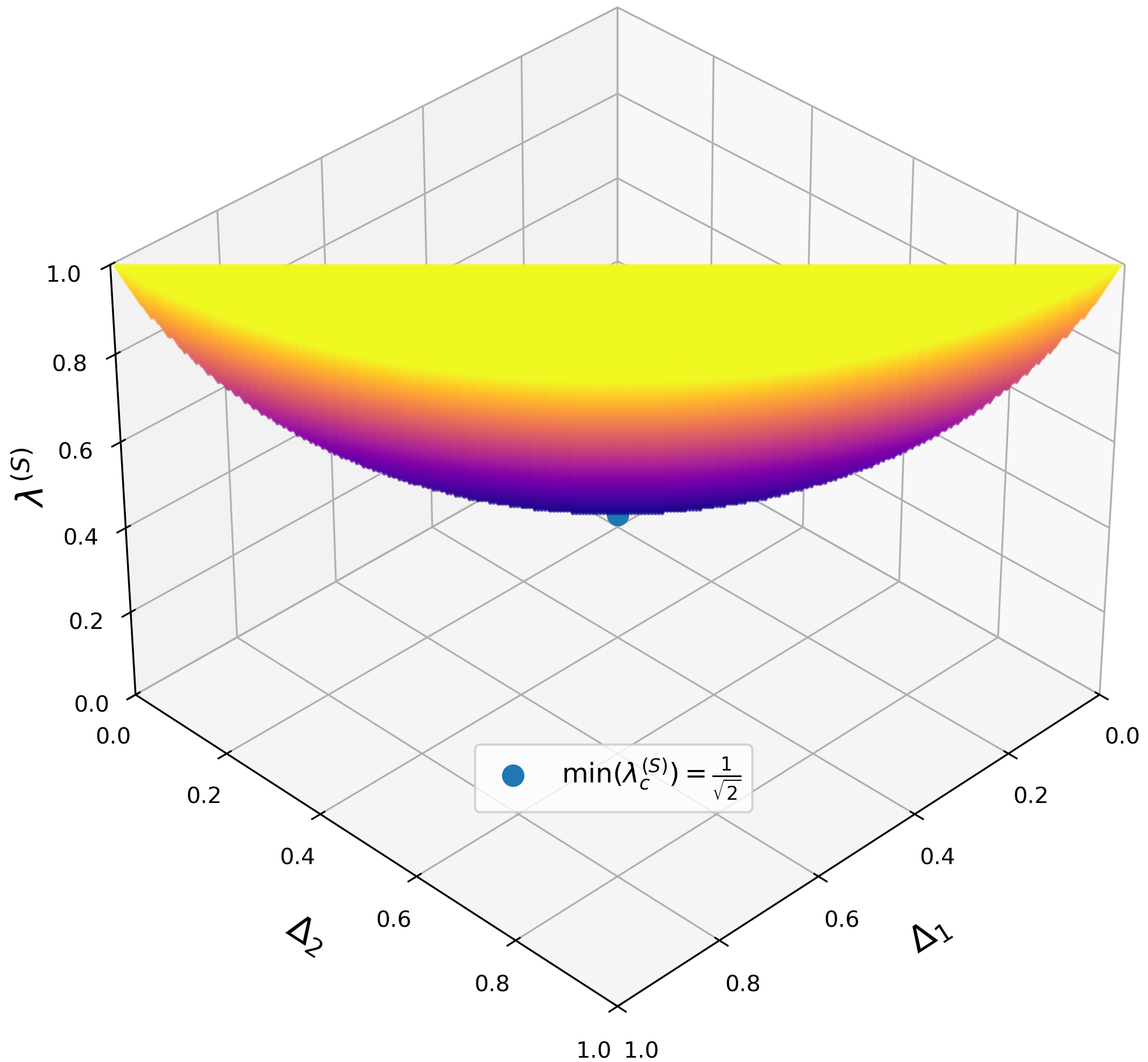}
        \caption{}
        \label{figlameq}
    \end{subfigure}
    \hfill
    \begin{subfigure}{0.48\textwidth}
        \centering
        \includegraphics[width=\linewidth]{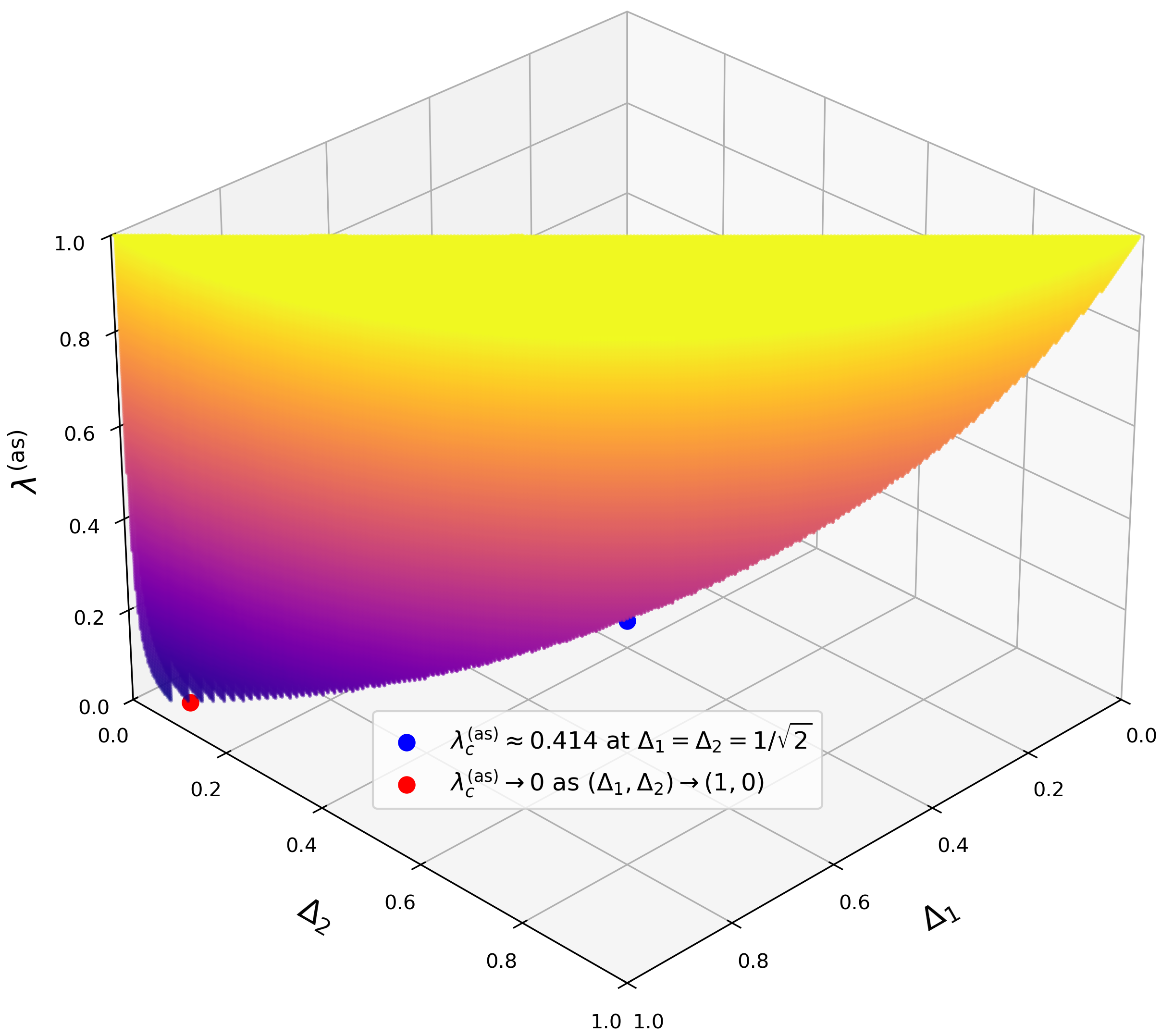}
        \caption{}
        \label{figlamsing}
    \end{subfigure}
    \caption{ Regions of the unsharpness parameter yielding a quantum advantage in the $2\to1$ RAC as functions of the preparation distinguishabilities $\Delta_1$ and $\Delta_2$, subject to the constraint $\Delta_1^2+\Delta_2^2\leq 1$. (a) Symmetric setting: the coloured region represents
    $\lambda^{(s)}>1/(\Delta_1+\Delta_2)$ that enables quantum advantage. The minimal threshold $\min(\lambda_c^{(s)})=1/\sqrt{2}$ is attained at $\Delta_1=\Delta_2=1/\sqrt{2}$.
    (b) Asymmetric setting $(\lambda_1=1)$: the coloured region corresponds to $\lambda^{(\rm as)}>\lambda_c^{(\rm as)} = (1-\Delta_1)/\Delta_2$ that enables quantum advantage. The blue point marks $\lambda_c^{(s)}\approx0.414$ at
    $\Delta_1=\Delta_2=1/\sqrt{2}$ (for which the minimal threshold of $\lambda_c^{(s)}$ is attained in the symmetric setting), while the red point corresponds to  $\min(\lambda_c^{(\rm as)})\to0$ in the limit
    $(\Delta_1,\Delta_2)\to(1,0)$.}
    \label{fig:lambda_regions}
\end{figure*}
%%%%%%%%%%%%%%%%%%%%%%%%%%%%%%%%%%%%%%%%%%%%%%%%%%%%%%%%%%%%%%%%%%%%%%%%%%%%%%%%%%%%%%%%%%%%%%%%%%%%%%%%%%%%%%%%%%%%%%%%%%%%%%%%%%%%%%%%%%%%%%%%%%%%%%%%%%%%%%%%%%%%%%%%%%%%%%%%%%%%%%%%%%%%%%%%%%%%%%%%%%%%%%%%%%%%%%%%%%%%%%%%%%%%%%%%%%%%%%%%%%%%%%%%%%%%%%%%%%%%%%%%%%%%%%%%%%%%%%%%%%%%%%%%%%%%%%%%%%%%%%%%%%%%%%%%%%%%%%%%%%%

\section{Evolution of Preparation Distinguishabilities under Sequential Measurements} \label{apcapd}

We analyse how preparation distinguishabilities evolve under the measurement channel applied sequentially by multiple receivers. Let us denote by $\Phi^{(k)}$ the non-selective measurement channel corresponding to the $k$-th receiver, Bob$^{(k)}$. If Bob$^{(k-1)}$ receives the ensemble of qubits $\qty{\rho_{x_1x_2}^{(k-1)}\in \mathscr{L}(\mathcal{H}^2)}$, then after the action of his measurement channel, the states passed on to Bob$^{(k)}$ are given by $\rho_{x_1x_2}^{(k)}=\Phi^{(k-1)}\qty(\rho_{x_1x_2}^{(k-1)})$. Let us assume the same set of $B_1$ and $B_2$ for any Bob$^{(k)}$, i.e.  the choice of the two measurements for each Bob$^{(k)}$ (for any $k = 1, 2, 3, \cdots$) is the following: sharp projective measurement of $B_1$, and unsharp measurement of $B_2$  parametrized by $\lambda_k$. With these, the channel $\Phi^{(k-1)}$ is defined as an equal mixture of two measurement channels (a dephasing channel and an unsharp measurement channel), and admits the Kraus representation
\begin{equation}\label{aptotchan}
   \Phi^{(k-1)}\qty(\rho_{x_1x_2}^{(k-1)})= \frac{1}{2}\qty[\Phi_1 \qty(\rho_{x_1x_2}^{(k-1)})+\Phi^{(k-1)}_2\qty(\rho_{x_1x_2}^{(k-1)})] = \frac{1}{2}\sum_{b \in \qty{\pm}}\qty(B_{b|1} \rho^{(k-1)}_{x_1x_2} B_{b|1}+ K^{(k-1)}_{b} \rho^{(k-1)}_{x_1x_2} K^{(k-1)}_{b}).
\end{equation}
An equal mixture of two channels is considered here as each Bob receives his input randomly. Here, $\Phi_1$ denotes the dephasing channel associated with a projective measurement of the observable $B_1$, which is independent of $k$, while $\Phi_2^{(k-1)}$ represents an unsharp measurement channel corresponding to the observable $B_2$. The projective measurement operators and the Kraus operators for the unsharp measurement are given by
\begin{equation}\label{bmg} 
B_{\pm 1|1} = \frac{1}{2}\qty(\openone \pm B_1), \quad  K_{\pm}^{(k)}=\alpha_k \openone \pm \beta_k B_2,
\end{equation}
where the parameters satisfy
\begin{equation}\label{bmgkradef} 
\alpha_k:=\frac{1}{2}\qty(\sqrt{\frac{1+\lambda_k}{2}}+\sqrt{\frac{1-\lambda_k}{2}}), \ \beta_k:=\frac{1}{2}\qty(\sqrt{\frac{1+\lambda_k}{2}}-\sqrt{\frac{1-\lambda_k}{2}}).
\end{equation}
Define the marginal ensembles for Bob$^{(k)}$
\begin{equation} \label{aliceprbobk}
\begin{aligned}
    \rho_{0}^{(y=1,k)} &= \frac{1}{2}\qty(\rho^{(k)}_{00}+\rho^{(k)}_{01}); \ \ \rho_{1}^{(y=1,k)} = \frac{1}{2}\qty(\rho^{(k)}_{10}+\rho^{(k)}_{11}); \\
    \rho_{0}^{(y=2,k)} &= \frac{1}{2}\qty(\rho^{(k)}_{00}+\rho^{(k)}_{10}); \ \ \rho_{1}^{(y=2,k)} = \frac{1}{2}\qty(\rho^{(k)}_{01}+\rho^{(k)}_{11}).
\end{aligned}
\end{equation} 
Upon receiving the input $y_k \in \{1,2\}$, the task of Bob$^{(k)}$ is to discriminate between the two above mentioned marginal ensembles with $y=y_k$. We fix $B_1$ and $B_2$ as the Helstrom measurements used by the first receiver, Bob$^{(1)}$, to discriminate the marginal ensembles $\qty{\rho_0^{(y,1)},\rho_1^{(y,1)}}$ for $y=1$ and $y=2$ respectively
\begin{equation}
   \rho_{0}^{(y,k=1)} -\rho_{1}^{(y,k=1)}  \propto B_{y} \, \, \forall \, \, y \in \{1,2\}.
   \label{helstrom1}
\end{equation}
In the following we show that the same $B_1$ and $B_2$ remain optimal for any subsequent Bob$^{(k)}$ to discriminate $\qty{\rho_0^{(y,k)},\rho_1^{(y,k)}}$ for $y=1$ and $y=2$ respectively if $\qty{B_1,B_2}=0$, even after the wave-function collapses due to measurements by the preceding Bobs.
\begin{Lemma}
Let the Helstrom measurements $B_{1}$ and $B_2$ are used by Bob$^{(1)}$ for discriminating the pairs $\qty{\rho_{0}^{(y,1)},\rho_{1}^{(y,1)}}$ for $y=1$ and $y=2$ respectively. If $\qty{B_1,B_2}=0$, the same measurements $B_{1}$ and $B_2$ remain the Helstrom measurements for discriminating the pair $\qty{\rho_{0}^{(y,k)},\rho_{1}^{(y,k)}}$ for $y=1$ and $y=2$, respectively, for all $k$, where 
 \begin{equation}
   \rho_{x}^{(y,k)}=\Phi^{(k-1)}\qty(\rho_{x}^{(y,k-1)}), \ \ \forall x \in \qty{0,1}, \ y\in\qty{1,2}  
 \end{equation}
and $\Phi^{(k-1)}$ is the non-selective measurement channel defined in Eq.~\eqref{aptotchan} and the above equation follows due to the linearity of the channel. Consequently, the preparation distinguishability can be expressed in terms of the trace norm as
 \begin{equation} \label{pdktnpre}
     \Delta_{y}^{(k)}=\frac{1}{2}\norm{\rho_{0}^{(y,k)} -\rho_{1}^{(y,k)}}_1  \, \, \forall \, \, y \in \{1,2\}.
 \end{equation}
\end{Lemma}

\begin{proof}
From Eq.~\eqref{helstrom1}, we have $\rho_{0}^{(y,1)}-\rho_{1}^{(y,1)} \propto B_{y}$. Hence, the Helstrom observable for discriminating the pair $\qty{\rho_{0}^{(y,1)},\rho_{1}^{(y,1)}}$ is given by $B_{y}$.  

Using the explicit form of the sub-channels, we evaluate how the operators $B_1$ and $B_2$ transform under each sub-channel. Using the identity $B_i B_{y} B_i=\qty{B_i, B_{y} } B_i-B_{y}$ for all $y, i=1,2$, we find
\begin{equation}
\begin{aligned}
 \Phi_1(B_{y})&=\frac{1}{2}\qty(B_{y} +B_1 B_{y} B_1)=\frac{1}{2}\qty{B_1,B_{y}}B_1, \\
 \Phi_2^{(k-1)}(B_{y})&=2\alpha_{k-1}^2 B_{y}+2\beta_{k-1}^2 B_2B_{y} B_2=2\qty(\alpha_{k-1}^2-\beta_{k-1}^2) B_{y} +2\beta_{k-1}^2 \qty{B_2,B_{y}}B_2.
\end{aligned}
\end{equation}
Therefore, for the full channel $\Phi^{(k-1)}=\frac{1}{2}\qty(\Phi_2^{(k-1)}+\Phi_1)$ action is given by
\begin{equation} \label{aprn1}
    \begin{aligned}
\Phi^{(k-1)}(B_1)&=\frac{1}{2}\qty(1+\sqrt{1-\lambda^2_{k-1}})B_1 + \beta_{k-1}^2 \qty{B_1,B_2}B_2, \\
\Phi^{(k-1)}(B_2)&=\frac{1}{2}B_2+ \frac{1}{4}\qty{B_1,B_2}B_1.
    \end{aligned}
\end{equation}
Let  $B_{y}$ be the Helstrom measurement the discrimination of the marginal ensembles $\qty{\rho_0^{(y,k-1)},\rho_1^{(y,k-1)}}$ by the $(k-1)$-th receiver (i.e. for Bob$^{(k-1)}$), then it must hold that $\rho_{0}^{(y,k-1)}-\rho_{1}^{(y,k-1)}=s_{y}^{(k-1)}B_{y}$ for some $s_{y}^{(k-1)}  \in \mathbb{R}$. By linearity of the channel, for all $y\in \{1,2\}$,
\begin{equation}
 \rho_{0}^{(y,k)}-\rho_{1}^{(y,k)} = \Phi^{(k-1)}\qty(\rho_{0}^{(y,k-1)}-\rho_{1}^{(y,k-1)}) =s_{y}^{(k-1)}\Phi^{(k-1)}\qty(B_{y})=\begin{cases}
     \text{For } y=1 \implies & s_1^{(k-1)} \Phi^{(k-1)}(B_1) \\ 
     \text{For } y=2 \implies & s_2^{(k-1)} \Phi^{(k-1)}(B_2)
 \end{cases}.
\end{equation}
Using Eq.~\eqref{aprn1}, we see that a sufficient condition for the difference between the two marginal ensembles to remain proportional to $B_{y}$ at every step, i.e. $\rho_{0}^{(y,k)}-\rho_{1}^{(y,k)} \propto B_{y}$, is $\qty{B_1,B_2}=0$. Under this condition, the total channel action is given by
\begin{equation}
 \Phi^{(k-1)}(B_1) =\frac{1}{2}\qty(1+\sqrt{1-\lambda_{k-1}^2}) B_1, \quad  \Phi^{(k-1)}(B_2) =\frac{1}{2}B_2.
\end{equation}
Hence, we obtain that for some $s_{y}^{(k)} \in \mathbb{R}$,
\begin{equation}
 \rho_{0}^{(y,k)}-\rho_{1}^{(y,k)} = s_{y}^{(k)}B_{y} \, \, \forall \, \, y \in \{1,2\}.  
\end{equation}
The trace norm expression of the preparation distinguishability given by Eq.~\eqref{pdktnpre} follows immediately.

\end{proof}

We now evaluate the evolution of the preparation distinguishabilities
\begin{equation}\label{appditstk}
    \Delta_{y}^{(k)}=\frac{1}{2}\norm{\rho_{0}^{(y,k)} - \rho_{1}^{(y,k)}}_1=  \frac{1}{2} \abs{s_{y}^{(k-1)}} \ \norm{\Phi^{(k-1)}\qty(B_{y})}_1 \ \ \forall s_{y}^{(k-1)} \in \mathbb{R}, y \in \{1,2\}.
\end{equation}
Using the structure derived above,
\begin{equation}\label{itpres}
\begin{cases}
    \text{For } y=1 \implies & \Delta_1^{(k)} = \frac{1}{4} \abs{s_1^{(k-1)}} \ \abs{1+\sqrt{1-\lambda_{k-1}^2}} \ \norm{B_1}_1 = \frac{1}{4}\abs{s_1^{(k-1)}}\qty(1+\sqrt{1-\lambda_{k-1}^2}) \ \norm{B_1}_1, \\
      \text{For } y=2 \implies & \Delta_2^{(k)} = \frac{1}{4} \abs{s_2^{(k-1)}} \ \norm{B_2}_1.
\end{cases}
\end{equation}
On the other hand, for the previous receiver (Bob$^{(k-1)}$), using $\rho_{0}^{(y,k-1)}-\rho_{1}^{(y,k-1)}=s_{y}^{(k-1)}B_{y}$, we find that
\begin{equation}\label{itpres0}
    \Delta_y^{(k-1)}=\frac{1}{2} \abs{s_y^{(k-1)}} \ \norm{B_y}_1 \, \, \forall \, \, y \in \{1,2\}.
\end{equation}
Now, from Eqs.~\eqref{itpres} and \eqref{itpres0}, we obtain the recursive relations
\begin{equation}\label{deltainduct1}
 \Delta_1^{(k)}= \frac{1}{2}\qty(1+\sqrt{1-\lambda_{k-1}^2}) \Delta_1^{(k-1)}, \quad \Delta_2^{(k)}= \frac{1}{2} \Delta_2^{(k-1)}.
\end{equation}

The action of the non-selective measurement channel (\ref{aptotchan}) on the initial states $\qty{\rho_{x_1x_2}}$ is presented in Fig. \ref{figbloconhel}.

\begin{figure}[t!]
\includegraphics[width=0.5\linewidth]{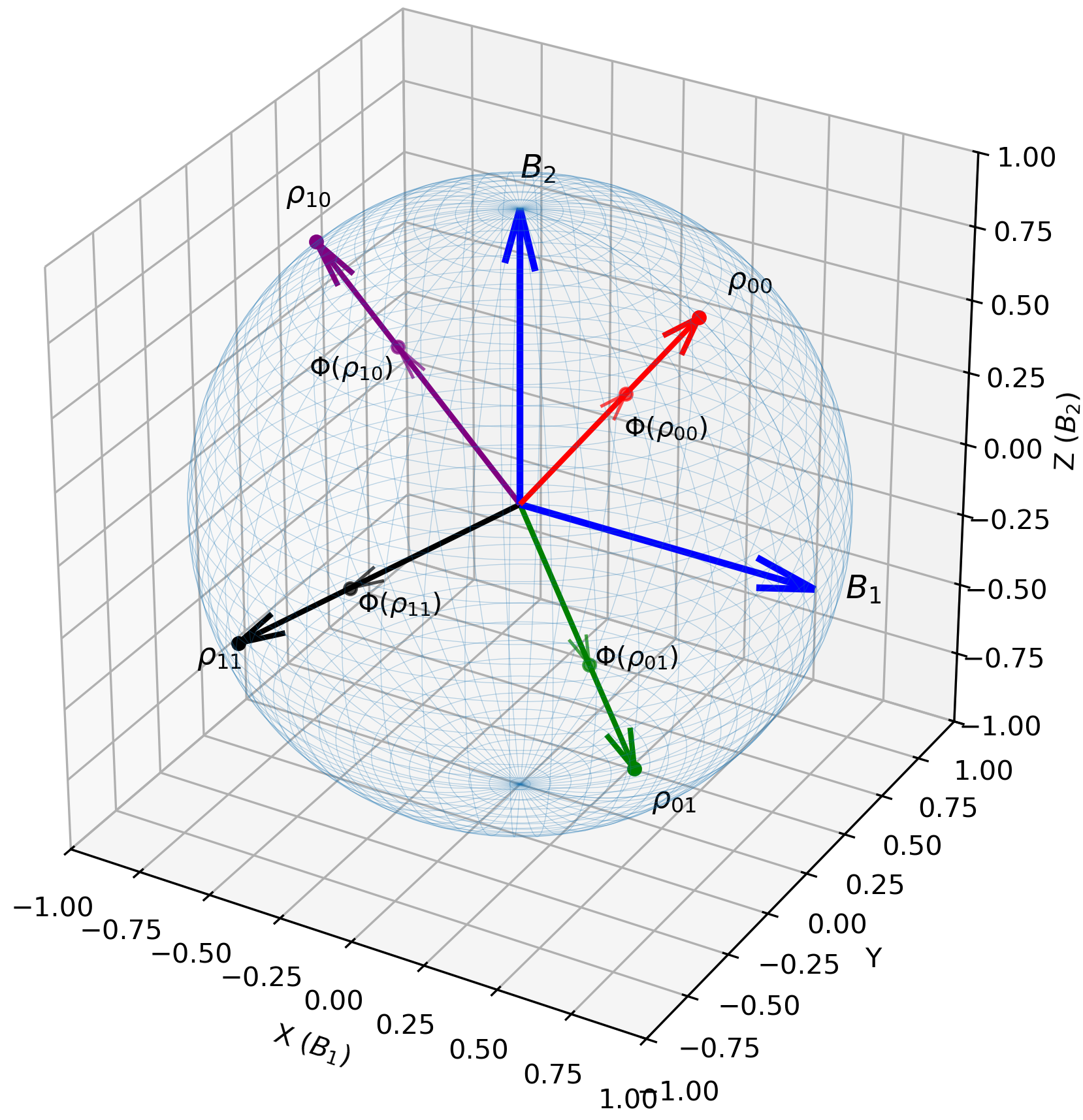}
\caption{Bloch sphere representation of the preparation states and measurement directions before and after the action of the non-selective measurement channel. The initial states $\qty{\rho_{x_1x_2}}$ and the transformed states $\qty{\Phi(\rho_{x_1x_2})}$  are shown as coloured dotted vectors, while the measurement observables $B_1$  and $B_2$ are indicated by solid arrows along the $X$- and $Z$-axes, respectively. The channel induces an anisotropic contraction of the Bloch vectors without altering their directions, thereby preserving the optimal (Helstrom) measurement basis while reducing the distinguishability.}
\label{figbloconhel}
\end{figure}

%%%%%%%%%%%%%%%%%%%%%%%%%%%%%%%%%%%%%%%%%%%%%%%%%%%%%%%%%%%%%%%%%%%%%%%%%%%%%%%%%%%%%%%%%%%%%%%%%%%%%%%%%%%%%%%%%%%%%%%%%%%%%%%%%%%%%%%%%%%%%%%%%%%%%%%%%%%%%%%%%%%%%%%%%%%%%%%%%%%%%%%%%%%%%%%%%%%%%%%%%%%%%%%%%%%%%%%%%%%%%%%%%%%%%%%%%%%%%%%%%%%%%%%%%%%%%%%%%%%%%%%%%%%%%%%%%%%%%%%%%%%%%%%%%%%%%%%%%%%%%%%%%%%%%%%%%%%%%%%%%%%

\section{Arbitrarily Large Sequence of Receivers achieving Quantum Advantages in $2 \to 1$ RAC} \label{apsharecal}

At first, we present the formal proof of the Theorem 3 mentioned in the main paper. 

\begin{proof}
We begin by recalling the evolution of the preparation distinguishabilities. Iterating the recursion relations given in Eq.~(\ref{deltainduct1}), we obtain
\begin{equation}\label{deltainduct2}
 \Delta_1^{(k)}= \frac{M_k}{2^{k-1}}\Delta_1^{(1)}, \quad \Delta_2^{(k)}= \frac{1}{2^{k-1}} \Delta_2^{(1)}, \quad \quad M_{k\geq 2}:=\prod_{l=1}^{k-1} \qty(1+\sqrt{1-\lambda_{l}^{2}}), \  M_1:=1.
\end{equation}
Without loss of generality, we parametrise the initial distinguishabilities as
\begin{equation}
\label{distingbob1}
 \forall r_1,r_2 \in [0,1], \omega \in \qty(0, \frac{\pi}{2}), \ \   \Delta^{(1)}_1:=r_1\cos\omega, \quad \Delta^{(1)}_2:=r_2\sin\omega, \ \ \text{with } \qty(\Delta^{(1)}_1)^2+\qty(\Delta^{(1)}_2)^2 \leq \max(r_1^2,r_2^2) \leq 1.
\end{equation}
Using the condition for quantum advantage, $\mathcal{P}^{(k)}_{\rm avg}>\frac{3}{4}$, we now establish the existence of a sequence $\qty{\lambda_k \in (0,1)}$ for the parameter $\omega \in (0,\frac{\pi}{2})$ that enables quantum advantages by an arbitrarily large sequence of independent receivers. It suffices to consider the case $r_1=1$, and $r_2=r\in (0,1]$. In this setting, a sufficient condition is
\begin{equation}\label{lambgenr1}
    \lambda_k > \frac{2^{k-1} -\cos\omega M_k}{r \sin\omega } \ \ \forall k \geq 2, \ \ \ \text{with } \lambda_1 > \frac{1-\cos\omega}{r \sin\omega}=\frac{1}{r}\tan\frac{\omega}{2}.
\end{equation}
Introducing a parameter $\epsilon > 0$, we define
\begin{equation}\label{lambseq}
\forall r \in (0,1], \ \ \lambda_{k}(\omega) := \begin{cases}
    (1+\epsilon)\frac{1}{r}\tan\frac{\omega}{2} & \text{for } k=1 \\
    (1+\epsilon)\frac{\qty(2^{k-1} -\cos\omega M_k)}{r \sin\omega} & \text{for } k \geq 2, \ 
\end{cases}, 
\end{equation} 
A sequence $\qty{\lambda_k}_{k=1}^{N}\subset(0,1)$ satisfying  the above recursion (Eq.~\ref{lambseq}) is called feasible. 

The feasibility of the sequence $\qty{\lambda_k}$ is established as follows. First, we show that whenever the sequence satisfies $\lambda_k<1 \ \forall k$, the recursive relation implies that the sequence is strictly monotonically increasing for all $0<\omega<\pi/2$ and $r,\epsilon>0$. Within this range of $\omega$, since $\lambda_1>0$, monotonicity guarantees that every subsequent element of the sequence is also positive, hence $\lambda_k>0 \ \ \forall k$. Finally, we address the remaining question: whether within this range of $\omega$, there exists a non-empty sub interval for which all $\lambda_k$ remain strictly below unity for any $r,\epsilon>0$.

Using Eq.~\eqref{lambseq}, we obtain
\begin{equation}
    \lambda_{k+1}-2 \lambda_k=\frac{(1+\epsilon)\cos\omega}{r \sin\omega}\qty(2M_k-M_{k+1}),
\end{equation}
Since $M_{k+1}=M_k\qty(1+\sqrt{1-\lambda_k^2})$, it follows that
\begin{equation}
    \lambda_{k+1}-2 \lambda_k=\frac{(1+\epsilon)\cos\omega M_k}{r \sin\omega}\qty(1-\sqrt{1-\lambda_k^2}),
\end{equation}
Now, consider that  $\lambda_j<1 \ \forall j\leq k$, hence $\sqrt{1-\lambda_k^2}<1$ and $M_k>0$. Further, assuming $\omega\in\qty(0,\pi/2)$, implying $\sin\omega>0$, $\cos\omega>0$, we find
\begin{equation}\label{monosttin}
    \forall k, r>0,\epsilon>0, \ \ \lambda_{k+1}>2 \lambda_k>\lambda_k.
\end{equation}
Note that from Eq.~\eqref{lambseq}, the assumed range of $\omega\in\qty(0,\pi/2)$ ensures the positivity of $\lambda_1$, and the strict monotonicity ensures the positivity of all $\lambda_k \ \ \forall k$ for any $r,\epsilon>0$. 

In the following we show that there exists a sub-interval of assumed $\omega\in(0,\pi/2)$ such that all $\lambda_k<1$ for any $r,\epsilon>0$. To show this, let us analyse the recursion in the small-angle approximations $\sin\omega \approx \omega + \order{\omega^3}$ and $\cos\omega\approx 1- \omega^2/2 + \order{\omega^4}$. In this regime, from Eq.~\eqref{lambseq}, we obtain
\begin{equation}
  \lambda_1 = c_1^{(1)}\omega + \order{\omega^3}, \ \  c_1^{(1)}:=\frac{1+\epsilon}{2r}>0.
\end{equation}
Assume inductively that, for all $j < k$, $\lambda_j=\order{\omega}$. Then $\lambda_j^2=\order{\omega^2}$, and consequently 
\begin{equation}
  1+\sqrt{1-\lambda_j^2}=2\qty(1-\frac{\lambda_j^2}{4})+\order{\lambda^4_j} = 2\qty(1-\frac{\lambda_j^2}{4})+\order{\omega^4}  
\end{equation}
Substituting into the definition of $M_k$ given by Eq.~\eqref{deltainduct2}, we find
\begin{equation} \label{mk1}
M_k=2^{k-1}\prod_{j=1}^{k-1}\qty[1-\frac{\lambda_j^2}{4}+\order{\omega^4}].
\end{equation}
Note that since $\lambda_j^2=\order{\omega^2}$, therefore $\lambda_j^2\lambda_i^2=\order{\omega^4}$. More generally, every product involving $p$ quadratic terms, $\prod_{m=1}^p \lambda_{j_m}^2$ is of order $\order{\omega^{2p}}$. Hence all higher-order mixed terms generated from the product expansion of $M_k$ contribute only at order $\order{\omega^4}$ or higher. Therefore, Eq.~\eqref{mk1} reduces to
\begin{equation}\label{mk2}
  M_k = 2^{k-1} \qty(1-\frac{1}{4}\sum_{j=1}^{k-1}\lambda_j^2 + \order{\omega^4})=2^{k-1} \qty(1-\frac{S_k}{4}) + \order{\omega^4}, \ \ S_k:=\sum_{j=1}^{k-1}\lambda_j^2, \ \ S_1=0. 
\end{equation}
Therefore,  we get from Eq.(\ref{lambseq})
\begin{equation} \label{darunlam}
    \lambda_k=2^{k-1}c_1^{(1)} \qty(\omega+\frac{S_k}{2 \omega})+\order{\omega^3}.
\end{equation}
Since $S_k=\order{\omega^2}$, it follows immediately that $\lambda_k=\order{\omega}$. This proves inductively that, for any $\epsilon >0$ and $r \in (0,1]$, one can always choose a $\omega \to 0^+$,  such that $\lambda_k$ becomes arbitrarily small ensuring $\lambda_k<1 \ \ \forall k$. This prove the feasibility of the sequence $\qty{\lambda_k}_{k=1}^N$.
\end{proof}

%%%%%%%%%%%%%%%%%%%%%%%%%%%%%%%%%%%%%%%%%%%%%%%%%%%%%%%%%%%%%%%%%%%%%%%%%%%%%%%%%%%%%%%%%%%%%%%%%%%%%%%%%%%%%%%%%%%%%%%%%%%%%%%%%%%%%%%%%%%%%%%%%%%%%%%%%%%%%%%%%%%%%%%%%%%%%%%%%%%%%%%%%%%%%%%%

\subsection{Estimating $\omega$ for a fixed number of receivers achieving quantum advantages}

The above result tells that there exists an $\omega$ for which an arbitrarily large sequence of receivers can each get a quantum advantage. In the following, we present a process to estimate the value of $\omega$ required for the sequential quantum advantage by $k$ receivers for very small values of $\omega$ and $\lambda_{j<k}$. Using Eqs.~\eqref{darunlam} and \eqref{lambseq}, it follows that
\begin{equation}
\lambda_2 \approx \frac{1+\epsilon}{r}\qty[\omega\qty(1+\frac{(c_1^{(1)})^2}{2})-\omega^3 \frac{(c_1^{(1)})^2}{4}]  = c_2^{(1)} \omega + c_2^{(3)} \omega^3,
\end{equation}
with
\begin{equation} \label{lambda2approxappendix}
c_2^{(1)}:=\frac{1+\epsilon}{r}\qty[1+\frac{(1+\epsilon)^2}{8r^2}] = 2 c^{(1)}_{1}+(c^{(1)}_{1})^{3}, \ \ c_2^{(3)}:=-\frac{(1+\epsilon)^3}{16r^3} = -\frac{(c^{(1)}_{1})^{3}}{2}.
\end{equation}
Now assume inductively that, for all $j<k$,
\begin{equation} \label{notunlamk1}
 \forall j<k, \ \lambda_{j}=c_{j}^{(1)} \omega + c_{j}^{(2)} \omega^2+ \cdots + c_{j}^{(2^{j}-1)} \omega^{2^{j}-1}=\sum_{m=1}^{2^{j}-1} c_j^{(m)}\omega^m .
\end{equation}
Using Eq.~\eqref{notunlamk1}, we obtain
\begin{equation}\label{notunlam2}
    S_k=\sum_{j=1}^{k-1} \qty(\sum_{m=1}^{2^{j}-1} c_j^{(m)}\omega^m)^2=\sum_{j=1}^{k-1}\sum_{m=1}^{2^{j}-1}\sum_{n=1}^{2^{j}-1} c_j^{(m)} c_j^{(n)} \omega^{m+n} 
\end{equation}
Hence, $S_k$ is itself a polynomial in $\omega$, whose lowest-order term is proportional to $\omega^2$ and highest power is $\omega^{2^{k}-2}$. Substituting this expression into Eq.~\eqref{lambseq} and using the small angle approximation, we obtain
\begin{align} \label{inductlamb}
    \lambda_k & \approx 2^{k-1}c^{(1)}_{1}\qty[\omega+\qty(\frac{1}{2\omega}-\frac{\omega}{4})\sum_{j=1}^{k-1}\sum_{m=1}^{2^{j}-1}\sum_{n=1}^{2^{j}-1} c_j^{(m)} c_j^{(n)} \omega^{m+n}]\\
    & = 2^{k-1}c^{(1)}_{1}\qty[\omega+\frac{1}{2}\sum_{j=1}^{k-1}\sum_{m=1}^{2^{j}-1}\sum_{n=1}^{2^{j}-1} c_j^{(m)} c_j^{(n)} \omega^{m+n-1}-\frac{1}{4}\sum_{j=1}^{k-1}\sum_{m=1}^{2^{j}-1}\sum_{n=1}^{2^{j}-1} c_j^{(m)} c_j^{(n)} \omega^{m+n+1}]
\end{align}
%For the first and the second term above, performing the transformations $(m-1) \to m$ and $m\to(m-1)$ respectively gives
%\begin{equation}
%\lambda_k \approx 2^{k-2}c^{(1)}_{1}\qty[\omega+\frac{1}{2}\sum_{j=1}^{k-1}\sum_{m=0}^{2^{j}-2}\sum_{n=1}^{2^{j}-1} c_j^{(m+1)} c_j^{(n)} %\omega^{m+n}-\frac{1}{4}\sum_{j=1}^{k-1}\sum_{m=2}^{2^{j}}\sum_{n=1}^{2^{j}-1} c_j^{(m-1)} c_j^{(n)} \omega^{m+n}]
%\end{equation}
%Defining $c^{(0)}_j=c^{(2^{j-1}+1)}_{j}=0$, w
We can rewrite this expression as 
\begin{equation}\label{weirdlamk1}
\lambda_k  \approx  2^{k-1}c^{(1)}_{1}\qty(1+\sum_{j=1}^{k-1} \frac{(c^{(1)}_{j})^{2}}{2}) \ \omega + O(\omega^{2})
\end{equation}
To obtain a rough estimate of the number of Bobs allowed by the upper bound on $\omega$, Eq.~\eqref{inductlamb} indicates that, for sufficiently small $\omega$, higher-order terms in $\omega^{\geq 2}$ may be neglected. Hence, 
\begin{equation}\label{weirdlamk1}
\lambda_k  \approx  2^{k-1}c^{(1)}_{1}\qty(1+\sum_{j=1}^{k-1} \frac{(c^{(1)}_{j})^{2}}{2}) \ \omega.
\end{equation}
In this approximation, the value of $\lambda_k$ depends only on the coefficient $c^{(1)}_j$. Next, our goal is to express each $c^{(1)}_j$ in terms of $c^{(1)}_1$, where $c_1^{(1)}:=(1+\epsilon)/(2r)$, such that all $\lambda_k$ can be expressed solely in terms of $\epsilon$, $r$ and $\omega$. Comparing the coefficients of $\omega$ in Eq.~\eqref{weirdlamk1} and \eqref{notunlamk1}, we obtain
\begin{equation}\label{lambfirstordercoeff}
c^{(1)}_{k} \approx 2^{k-1}c^{(1)}_{1}\qty(1+\sum_{j=1}^{k-1} \frac{(c^{(1)}_{j})^{2}}{2}) \ \ \forall \, k \geq 2, \ \ \text{with } c_1^{(1)}:=\frac{1+\epsilon}{2r}.
\end{equation}
To obtain an estimate of $c^{(1)}_{k}$, We now establish the following result.
\begin{Lemma}\label{prop:odd}
For every $k \geq 2$, the coefficient $c^{(1)}_k$ is a polynomial in $c^{(1)}_{1}$ containing only odd powers, i.e.
\begin{equation}
    c^{(1)}_k = \sum_{n=0}^{2^{k-1}-1} b_n^{(k)}\qty(c^{(1)}_{1})^{2n+1}, 
    \end{equation}
\end{Lemma}

\begin{proof}
For $k=2$,  from Eq.(\ref{weirdlamk1})
\begin{equation}
c^{(1)}_{2}=2 c^{(1)}_{1}+(c^{(1)}_{1})^{3},
\label{c1lemma}
 \end{equation}
which is consistent with Eq.(\ref{lambda2approxappendix}). The above  is clearly a polynomial containing only odd powers of $c^{(1)}_{1}$. Now, assume inductively that $c^{(1)}_j$ contains only odd powers of $c^{(1)}_1$ for all $1 \leq j \leq k-1$. Then each $c^{(1)}_j$ is an odd function of $c^{(1)}_1$, and consequently $(c^{(1)}_j)^{2}$ is an \emph{even} function containing only even powers of $c^{(1)}_1$. Therefore, the sum $\sum_{j=1}^{k-1}\frac{(c^{(1)}_j)^2}{2}$ contains only even powers of $c^{(1)}_1$, and adding $1$ yields a polynomial in $(c^{(1)}_1)^2$ with a nonzero constant term. Multiplying by $2^{k-1}c^{(1)}_1$ then shifts all powers by one, producing only odd powers of $c^{(1)}_1$. Hence $c^{(1)}_k$ is a polynomial in $c^{(1)}_1$ containing exclusively odd powers, completing the induction.

From Eq.(\ref{lambfirstordercoeff}), it follows that the lowest power of $c^{(1)}_{1}$ in $c^{(1)}_{k}$ is one. We now prove that the highest power of $c^{(1)}_{1}$ in $c^{(1)}_{k}$ is $2^k-1$, which will complete the proof of the lemma. From Eq.(\ref{c1lemma}), for $k=2$, the highest power of $c^{(1)}_{1}$ is $2^k-1=3$. Let us assume that the highest power of $c^{(1)}_{1}$ in $c^{(1)}_{k}$ is $2^k-1$ for some $k \geq 2$. Now, from Eq.(\ref{lambfirstordercoeff}), we have 
\begin{equation}\label{lambfirstordercoeff2}
c^{(1)}_{k+1} \approx 2^{k}c^{(1)}_{1}\qty(1+\sum_{j=1}^{k} \frac{(c^{(1)}_{j})^{2}}{2}),
\end{equation}
in which the highest power of $c^{(1)}_{1}$ will be contributed from the term $c^{(1)}_{1} (c^{(1)}_{k})^{2}/2$. As assumed, the highest power of $c^{(1)}_{1}$ in $c^{(1)}_{k}$ is $2^k-1$. Hence, the highest power of $c^{(1)}_{1}$ in $c^{(1)}_{1} (c^{(1)}_{k})^{2}/2$ is $2(2^k-1)+1 = 2^{k+1}-1$. This completes the proof.
\end{proof}
In view of Lemma~\ref{prop:odd}, we write
\begin{equation}\label{eq:ansatz}
    c^{(1)}_j=2^{j-1}c^{(1)}_1\,P_{j}(x), \qquad x := (c^{(1)}_1)^2,
\end{equation}
where $P_{j}(x) = \sum_{n=0}^{2^{j-1}-1} a_n^{(j)}\,x^n$ is a polynomial in $x$. We take $P_1(x)=1$ corresponding to the convention $c^{(1)}_{1}=c^{(1)}_{1}P_1(x)$. Substituting this into into Eq.~(\ref{lambfirstordercoeff}), for $k\geq 2$ we find
\begin{equation}\label{eq:Prec}
    c_{k}^{(1)}=2^{k-1}c^{(1)}_{1}P_{k}(x)
    = 2^{k-1}c^{(1)}_{1}\left(1 + \frac{x}{2}
       + \sum_{j=2}^{k-1}\frac{2^{2(j-1)}xP_{j}(x)^2}{2}\right) \implies  P_{k}(x)=1 + \frac{x}{2} + \sum_{j=2}^{k-1} 2^{2j-3}\,x\,P_{j}(x)^2.
\end{equation}
Using the expression for $P_{k-1}(x)$, this recurrence relation may be rewritten as
\begin{equation}\label{eq:step}
P_k(x) =P_{k-1}(x) + 2^{2k-5}\,x\,P_{k-1}(x)^2,
\end{equation}
which holds for $k > 2$, with initial conditions $P_1(x) =1$ and $P_{2}(x)=1+\frac{x}{2}$. To verify consistency, let us compute the first few terms
\begin{equation}
\begin{aligned}
    P_3(x) &= P_2 + 2xP_{2}^2
            = \left(1+\tfrac{x}{2}\right)+2x\left(1+x+\tfrac{x^2}{4}\right)
            = 1 + \frac{5x}{2} + 2x^2 + \tfrac{x^3}{2},\\
    P_4(x) &= P_3 + 2^{3}x P_{3}^2= 1 +\frac{21x}{2} + 42x^{2} + \frac{165x^{3}}{2} +88x^{4} + 52x^{5} + 16x^{6} + 2x^{7}.\\
\vdots
\end{aligned}
\end{equation}
These reproduce the explicit expressions obtained by direct computation
\begin{equation}\label{somcoeff}
\begin{aligned}
    c^{(1)}_{2} &= 2 c^{(1)}_{1}+(c^{(1)}_{1})^{3}, \\
    c^{(1)}_{3} &= 4c^{(1)}_{1} + 10(c^{(1)}_1)^{3} + 8(c^{(1)}_1)^5 + 2(c^{(1)}_1)^7,\\
    c^{(1)}_{4} &= 8c^{(1)}_{1} + 84(c^{(1)}_1)^{3} + 336(c^{(1)}_{1})^{5} + 660(c^{(1)}_1)^{7} + 704(c^{(1)}_1)^{9} + 416(c^{(1)}_1)^{11} + 128(c^{(1)}_1)^{13} + 16(c^{(1)}_1)^{15}, \\ 
    \vdots
\end{aligned}
\end{equation}
It should be noted that the bound on $\lambda_k(\omega)$ obtained in Eq.~\eqref{weirdlamk1} is not tight, since higher-order contributions in $\omega$ have been neglected. A more accurate bound could be obtained by retaining higher-order terms in the perturbative expansion. Nevertheless, the present estimate provides a useful guideline for constructing a sequential scenario exhibiting quantum advantage, and for selecting suitable values of $r,\epsilon,\omega$. 

Now, suppose we aim to achieve quantum advantages for up to $k$ Bobs within the linear approximation for $\lambda_j$. The final Bob (Bob$^{(k)}$) can choose $\lambda_{k}\approx c^{(1)}_{k}\omega=1$, where higher order terms of $\omega$ have been neglected. From it, we can choose $\omega = 1/c^{(1)}_{k}$, where $c^{(1)}_{k}$ is expressed in terms of $c^{(1)}_{1} = (1 + \epsilon)/(2 r)$ in (\ref{somcoeff}). Hence, choosing $\epsilon \to 0$ and $r \leq 1$, one can estimate $\omega$ for this purpose.

As an illustrative example, let us choose $r=1$. Suppose that we wish to achieve quantum advantages for up to four Bobs ($k=4$). Setting $\lambda_{4}\approx c^{(1)}_{4}\omega=1$, Eq.~\eqref{somcoeff} gives
\begin{equation}
\omega=\frac{1}{ 8c^{(1)}_{1} + 84(c^{(1)}_1)^{3} + 336(c^{(1)}_{1})^{5} + 660(c^{(1)}_1)^{7} + 704(c^{(1)}_1)^{9} + 416(c^{(1)}_1)^{11} + 128(c^{(1)}_1)^{13} + 16(c^{(1)}_1)^{15}}.
\end{equation}
Assuming $\epsilon \ll 1$ and retaining only linear terms, we approximate $(c^{(1)}_{1})^{n}\approx (1+n\epsilon)/2^n$, which yields
\begin{equation}
\omega \approx \frac{2048}{85\qty(765+3347\epsilon)} \xRightarrow{\epsilon=10^{-4}} \omega \approx 0.0315 \gtrsim \frac{1}{2^{5}}.
\end{equation} 
Thus, for $\omega\sim 2^{-5}$, one finds that at least four Bobs can simultaneously and independently achieve quantum advantage, i.e. $\mathcal{P}_{avg}^{(k)}>\frac{3}{4} \ \ \forall k \in \{1,2,3,4\}$.

\end{document}